\documentclass[aps,amsmath,amssymb,twocolumn,prl,superscriptaddress]{revtex4-2}

\usepackage{graphicx}
\usepackage{dcolumn}
\usepackage{bm}
\usepackage{epsfig,hyperref,amsthm}
\usepackage{mathtools}
\usepackage{color}
\usepackage{colordvi}
\usepackage[dvipsnames]{xcolor}
\usepackage{relsize}
\usepackage{accents}
\usepackage{wasysym}  
\usepackage{xypic}
\usepackage{times}
\usepackage{newtxmath}
\usepackage{amsfonts}

\hypersetup{
	colorlinks=true,
	linkcolor=CitingColor,
	citecolor=CitingColor,
	urlcolor=CitingColor
}

\DeclareMathAlphabet\mathbfcal{OMS}{cmsy}{b}{n}
\newcommand{\ket}[1]{\ensuremath{|#1\rangle}}
\newcommand{\bra}[1]{\ensuremath{\langle #1|}}

\newcommand{\proj}[1]{\ket{#1}\bra{#1}}
\newcommand{\be}{\begin{equation}}
\newcommand{\ee}{\end{equation}}
\newcommand{\ba}{\begin{eqnarray}}
\newcommand{\ea}{\end{eqnarray}}

\newcommand{\norm}[1]{\left\|#1\right\|}

\newcommand{\id}{\mathbb{I}}

\newcommand{\freeset}{\mathcal{F}_{R}}
\newcommand{\freeop}{\mathcal{O}_R}

\newcommand{\Wdiff}{\Delta}

\newtheorem{result}{Theorem}

\newtheorem{result-coro}[result]{Corollary}
\newtheorem{lemma}[result]{Lemma}

\newtheorem{question}{Question}

\definecolor{nred}{rgb}{0.9,0.1,0.1}
\definecolor{nblack}{rgb}{0,0,0}
\definecolor{nblue}{rgb}{0.2,0.2,0.8}
\definecolor{ngreen}{rgb}{0.2,0.6,0.2}
\definecolor{ublue}{rgb}{0,0,0.5}

\definecolor{pur}{rgb}{0.75,0,0.75}
\definecolor{nngrn}{rgb}{0,0.5,0.5}
\definecolor{CitingColor}{rgb}{0,0.3,1}

\newcommand{\blu}{\color{nblue}}

\newcommand{\CY}[1]{{#1}}

\newcommand{\CYthree}[1]{{#1}}
\newcommand{\CYnew}[1]{{#1}}

\begin{document}
\title{
Complete Characterisation of State Conversions by Work Extraction}

\author{Chung-Yun Hsieh}
\email{chung-yun.hsieh@bristol.ac.uk}
\affiliation{H. H. Wills Physics Laboratory, University of Bristol, Tyndall Avenue, Bristol, BS8 1TL, United Kingdom}

\author{Manuel Gessner}
\email{manuel.gessner@uv.es}
\affiliation{Instituto de Física Corpuscular (IFIC), CSIC‐Universitat de València and Departament de Física Teòrica, UV, Avinguda Vicent Andrés Estellés 19, E-46100 Burjassot (Valencia), Spain} 

\date{\today}

\begin{abstract}
We introduce a thermodynamic work extraction task that describes the energy storage enhancement of quantum systems.
This task induces majorisation-like conditions that provide a necessary and sufficient characterisation of state conversions in general quantum resource theories. When applied to specific resources, these conditions reduce to the majorisation conditions under unital channels and provide a thermodynamic version of Nielsen's theorem in entanglement theory. We show how this result establishes the first universal resource certification class based on thermodynamics, and how it can be employed to quantify general quantum resources based on work extraction.
\end{abstract}

\maketitle

{\bf\em Introduction---}Quantum advantages underpin the development of quantum science and technologies. Enabled by {quantum resources}, it is possible for certain quantum information tasks to surpass the performance of classical strategies that do not make use of \CY{these resources~\cite{Kuroiwa2024PRL,Kuroiwa2024PRA,Meier2025PRXEnergy}}. For instance, quantum teleportation~\cite{Bennett93}, super-dense coding~\cite{Bennett92}, and sub-shot-noise interferometric precision with qubit probes~\cite{RevModPhys.90.035005,PhysRevLett.126.080502} are possible only with entanglement as a resource~\cite{HorodeckiRMP}. Different levels of advantages in so-called device-independent quantum information tasks in cryptography and communication are enabled by utilising quantum nonlocality~\cite{Brunner2014RMP,Acin2007PRL}, quantum steering~\cite{UolaRMP2020,Cavalcanti2016,Branciard2012PRA}, measurement incompatibility~\cite{Otfried2021Rev}, and quantum complementarity~\cite{Hsieh2023}. Furthermore, certain quantum dynamical features are, in fact, indispensable resources for quantum memories~\cite{Rosset2018PRX,Yuan2021npjQI,Ku2022PRXQ,Vieira2024,Abiuso2024,NarasimhacharPRL2019,Hsieh2025PRA-3}, quantum communication~\cite{Takagi2020PRL,Hsieh2021PRXQ,Hsieh2025PRL,Hsieh2025PRA}, preserving and generating quantum phenomena~\cite{Hsieh2020,Hsieh2021PRXQ,Liu2019DRT,Liu2020PRR,Hsieh2020PRR,Streltsov2015PRL}, and non-equilibrium thermodynamics~\cite{Hsieh2021PRXQ,Hsieh2025PRL,Hsieh2025PRA,Stratton2023}.

Most of the above-mentioned tasks and \CY{their} quantum advantages are resource-{\em dependent}. For instance, we do not expect a non-equilibrium state (a thermodynamic resource called athermality~\cite{Lostaglio2019,Faist2015NJP,HorodeckiPRL2003,HorodeckiPRA2003}) to be also useful for teleportation. Also, a quantum dynamics that is able to generate entanglement does not necessarily have a strong ability to transmit information. A natural question is thus whether there is a single, \CY{resource-{\em independent}} class of operational tasks in which general quantum resources can provide \CY{advantages} over resource-free ones. Performing such a class of tasks in the laboratory implies the ability to certify quantum resources operationally independent of the types of resources.

Remarkably, such {\em universal resource certification classes} \CY{(URCCs)} do exist. Through a general approach called {\em quantum resource theories}~\cite{ChitambarRMP2019}, \CY{or simply {\em resource theories},} it has been proved that general quantum resources \CY{(with reasonable physical assumptions)} can provide advantages in discrimination tasks~\cite{Piani2009PRL,Takagi2019PRL,Skrzypczyk2019PRL,Bae2019PRL,Takagi2019,Skrzypczyk2019,Hsieh2023-2,Hsieh2022PRR}, \CY{exclusion tasks~\cite{Ducuara2020PRL,Hsieh2023,Uola2020PRL}, parameter-estimation tasks in metrology~\cite{Tan2021PRL}, and input-output games~\cite{Uola2020PRA}. These URCCs tell us how to universally certify quantum resources via their operational advantages in the corresponding physical tasks~\footnote{\CY{See also Ref.~\cite{Rosset2020PRL} for spacelike separated resources}}, serving as vital interdisciplinary bridges.}

\CY{Surprisingly, there is no known thermodynamic URCC.} In thermodynamics, quantum signatures have been identified in quantum heat engines~\cite{KosloffPRE2002,FeldmannPRE2006,JiPRL2022,BeyerPRL2019,ChanPRA2022,Biswas2025PRL} and conservation laws~\cite{Jennings2010PRE,LostaglioNJP2017,Majidy2023,YungerHalpern2016NC,Guryanova2016NC} (see also Refs.~\cite{Lostaglio2020PRL,Levy2020PRXQ,Puliyil2022PRL,Upadhyaya2023,Centrone2024}), and thermodynamic advantages have been demonstrated by using entanglement~\cite{Lipkabartosik2023,Perarnau-LlobetPRX2015,Jennings2010PRE,Rio2011,Skrzypczyk2014NC}, coherence~\cite{Shiraishi2023,Korzekwa2016NJP}, steering~\cite{JiPRL2022,BeyerPRL2019,ChanPRA2022,Biswas2025PRL}, incompatibility~\cite{Hsieh2024}, and quantum dynamics~\cite{Hsieh2020PRR,Hsieh2021PRXQ,Hsieh2020,Hsieh2025PRL,Hsieh2025PRA}. Recently, it has been shown that one can witness certain quantum properties by observing heat~\cite{deOliveiraJunior2025PRL,Macedo2026}. Yet, it remains unknown whether there is a {\em single} class of operational work-extraction tasks in thermodynamics that can certify general quantum resources. Any such URCC, once found, can reveal operational advantages of general quantum resources in extracting controllable energy.

\CY{There is a {\em stronger version} of the above question---one whose answer would directly imply the existence of such a thermodynamic URCC. This stronger question is: 
\begin{center}
{\em Is there a class of work-extraction tasks that can completely characterise state conversions in general resource theories?}
\end{center}
Such a class, if exists, can provide a full ``ordering'' for the resource contents of states, which can be used to witness the resource~\cite{Buscemi2012PRL}. Moreover, such a class is by construction fully operational, as {\em work} is energy in a controlled, reusable form. 
}

Here, we fully answer the above question in the positive.
We introduce a work extraction task describing quantum system energy storage enhancement [Eq.~\eqref{Eq:Work diff}]. Then, we show that it can induce majorisation-like conditions to {\em completely} determine the (one-shot) state conversion in general resource theories (Theorem~\ref{Result:conversion}). 
Interestingly, when we apply this result to the resource theory of informational non-equilibrium (see, e.g., Ref.~\cite{Purity-review}), we reproduce the well-known majorisation conditions for state conversions under unital channels (Theorem~\ref{coro}). Finally, this result also allows us to obtain the first work-like thermodynamic URCC (Theorem~\ref{coro:state}) as well as to thermodynamically quantify general resources (Theorem~\ref{result:quantification R}).

An extension of our result beyond state resources, as well as its thermodynamic implications for steering and measurement incompatibility are discussed in a separate companion paper~\cite{Companion-arXiv}.

\begin{figure}
\scalebox{0.8}{\includegraphics{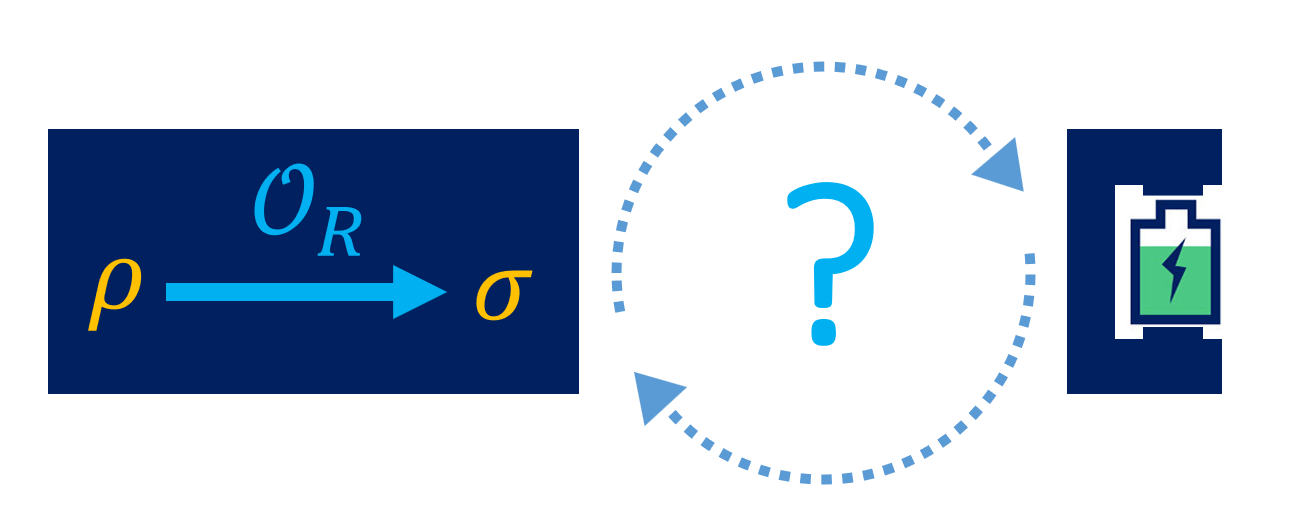}}
\caption{
{\bf Schematic illustration of the central question.}
We ask whether any class of work extraction tasks can completely characterise state conversions in general resource theories. 
A suitable answer can offer a method to analyse a broad range of quantum effects via work extraction.
}
\label{Fig:main_question}
\end{figure}

{\bf\em Quantum resource theories---}We begin by fixing notation. 
A quantum system is described by a state $\rho$, i.e., a positive semi-definite operator with unit trace: $\rho\ge0$ and ${\rm tr}(\rho)=1$.
The general evolution of such a system is represented by a quantum channel, a completely positive, trace-preserving linear map $\mathcal{E}$ that sends an initial state $\rho$ to $\mathcal{E}(\rho)$~\cite{QIC-book}.

To study diverse quantum properties on a common footing, one employs the framework of resource theories~\cite{ChitambarRMP2019}, which provide a unified formalism for characterising and quantifying quantum resources. To formalize the notion of a resource, let $R$ denote a quantum property of interest (e.g., entanglement or coherence).
A resource theory of $R$ is specified by a pair $(\freeset,\freeop)$. The set $\freeset$ consists of the {\em free states}, i.e., states that do not contain $R$.
A state $\rho$ is {\em resourceful} if $\rho\notin\freeset$.
For example, for entanglement, $\freeset$ is the set of separable states.

The set $\freeop$ is a set of channels, whose members are called {\em allowed operations.}
These are operations allowed to manipulate states.
Formally, any channel $\mathcal{N}$ in $\freeop$ must satisfy
\begin{align}\label{Eq: golden rule}
\mathcal{N}(\eta)\in\freeset \quad \forall\;\eta\in\freeset;
\end{align}
namely, it cannot generate $R$ from free inputs.
Equation~\eqref{Eq: golden rule} is a necessary condition for any allowed operation to satisfy, while one may consider extra constraints on $\freeop$ depending on the physical context; hence, generally, $\freeop$ does not contain all channels satisfying Eq.~\eqref{Eq: golden rule}.
Different choices of $\freeop$ lead to different operational settings for resource manipulation. For instance, in the case of entanglement, setting $\freeop$ as the set of local operations with classical communication (LOCCs) or as the set of local operations with shared randomness leads to different characterisations.

{\bf\em Working Hypotheses---}In most cases, $\freeop$ carries the following properties:
\begin{enumerate}
\item {\em Convexity.} Classical mixing of states is allowed.
\item {\em Compactness.} If a channel is arbitrarily close to an allowed operation, it is also allowed~\footnote{We always assume $\freeop$ is compact in the topology induced by the diamond distance, i.e., the distance measure induced by diamond norm $\norm{\cdot}_\diamond$~\cite{Watrous-book} (see also, e.g., Ref.~\cite{Regula2021Quantum}).}.
\item {\em Containing identity channel}. Doing nothing is allowed.
\item {\em Closedness under function composition}. Sequentially applying two allowed operations is also allowed.
\item {\em Containing state-preparation channels of free states}. That is, preparing free states is allowed.
\end{enumerate}
From now on, we assume that $\freeop$ describes a set of allowed operations with these properties. Notably, \CYthree{with this assumption}, as detailed in Appendix A, the free set $\freeset$ must be convex and compact (in the topology induced by the one norm).

{\bf\em Conversion problem in resource theories---}One of the central goals in arguably every resource theory is to understand the conversion under allowed operations.
To formalise it, consider a given $\freeop$ and two states $\rho$ and $\sigma$.
If there exists an allowed operation $\mathcal{N}\in\freeop$ achieving \mbox{$\mathcal{N}(\rho)=\sigma$}, we write $\rho\stackrel{\freeop}{\longrightarrow}\sigma$.
The {\em conversion problem} refers to the question:
\begin{center}
{\em Under which conditions do we have $\rho\stackrel{\freeop}{\longrightarrow}\sigma$?}
\end{center}
Conceptually, a complete answer to the conversion problem can tell us, necessarily and sufficiently, whether an evolution from one state to another via an allowed operation can ever be possible, just like the Second Law of thermodynamics (see, e.g., Refs.~\cite{Brandão2015,Ćwikliński2015PRL,Gour2018NC,Theurer2023NJP,Gour2022PRXQ}). Moreover, such an answer may serve as a bridge to further demonstrate the resource's advantage in specific operational tasks (e.g., as in Ref.~\cite{Buscemi2012PRL}).
Here, we aim to develop a fully thermodynamic way to answer conversion problems for general resource theories \CYthree{(Fig.~\ref{Fig:main_question})}. To this end, we need to introduce the following task.

{\bf\em Energy storage enhancement as a work extraction task---}In order to define a work extraction task that is capable of certifying quantum resources, we focus on a finite-dimensional quantum system with Hamiltonian $H$.
When this system reaches thermal equilibrium with a large bath in temperature $0~<~T~<~\infty$, it will be described by the thermal state
\begin{align}\label{Eq:thermal state}
\gamma = e^{-\frac{H}{k_BT}}/{\rm tr}\left(e^{-\frac{H}{k_BT}}\right),
\end{align}
where $k_B$ is the Boltzmann constant.
When this system is prepared in a non-equilibrium state $\rho$, one can extract the following optimal amount of work from it~\cite{Skrzypczyk2014NC,Brandao2013PRL,Footnote1}:
\begin{align}\label{Eq:W def}
W(\rho,H) \coloneqq (k_BT\ln2)D\left(\rho\,\|\,\gamma\right),
\end{align}
where 
$D(\rho\,\|\,\sigma)\coloneqq{\rm tr}\left[\rho\left(\log_2\rho - \log_2\sigma\right)\right]$ is the quantum (Umegaki) relative entropy~\cite{Umegaki1962}.
Notably, when the Hamiltonian is fully degenerate; i.e., $H=0$, the above equation gives the optimal work extractable from $\rho$'s information content:
\begin{align}\label{Eq:Winf def}
\CY{W_{{\rm inf}}(\rho) \coloneqq W(\rho,H=0) = (k_BT\ln2)D\left(\rho\,\|\,\id/d\right).}
\end{align}

Inspired by Refs.~\cite{Oppenheim2002PRL,Alimuddin2020PRE,Alimuddin2019PRA,Vaccaro2008PRA,Kwon2018PRL} and the charging process of quantum batteries~\cite{Perarnau-LlobetPRX2015,Vinjanampathy2016CP,Ciampini2017npjQI,Francica2017npjQI,Andolina2019PRL,Monsel2020PRL,Opatrny2021PRL,Yang2023PRL}, we consider the figure-of-merit called {\em energy storage enhancement} of $\rho$ with $H$:
\begin{align}\label{Eq:Work diff}
\Wdiff(\rho,H)\coloneqq W(\rho,H) - W_{{\rm inf}}(\rho).
\end{align}
$\Wdiff$ has a direct operational interpretation.
If the system is initially prepared in a state $\rho$ under a fully degenerate Hamiltonian, then $\Wdiff(\rho,H)$ is the change (possibly negative) in extractable work when the Hamiltonian is suddenly quenched to $H$ without altering the state (see Fig.~\ref{Fig}).
Concretely, the extractable work with the initial fully degenerate Hamiltonian is $W_{\rm inf}(\rho)$ [Eq.~\eqref{Eq:Winf def}], and after the quench to $H$ it is $W(\rho,H)$ [Eq.~\eqref{Eq:W def}].
The difference, $\Delta(\rho,H)$, quantifies the additional amount of energy that can be stored through this process.

\begin{figure}
\scalebox{0.8}{\includegraphics{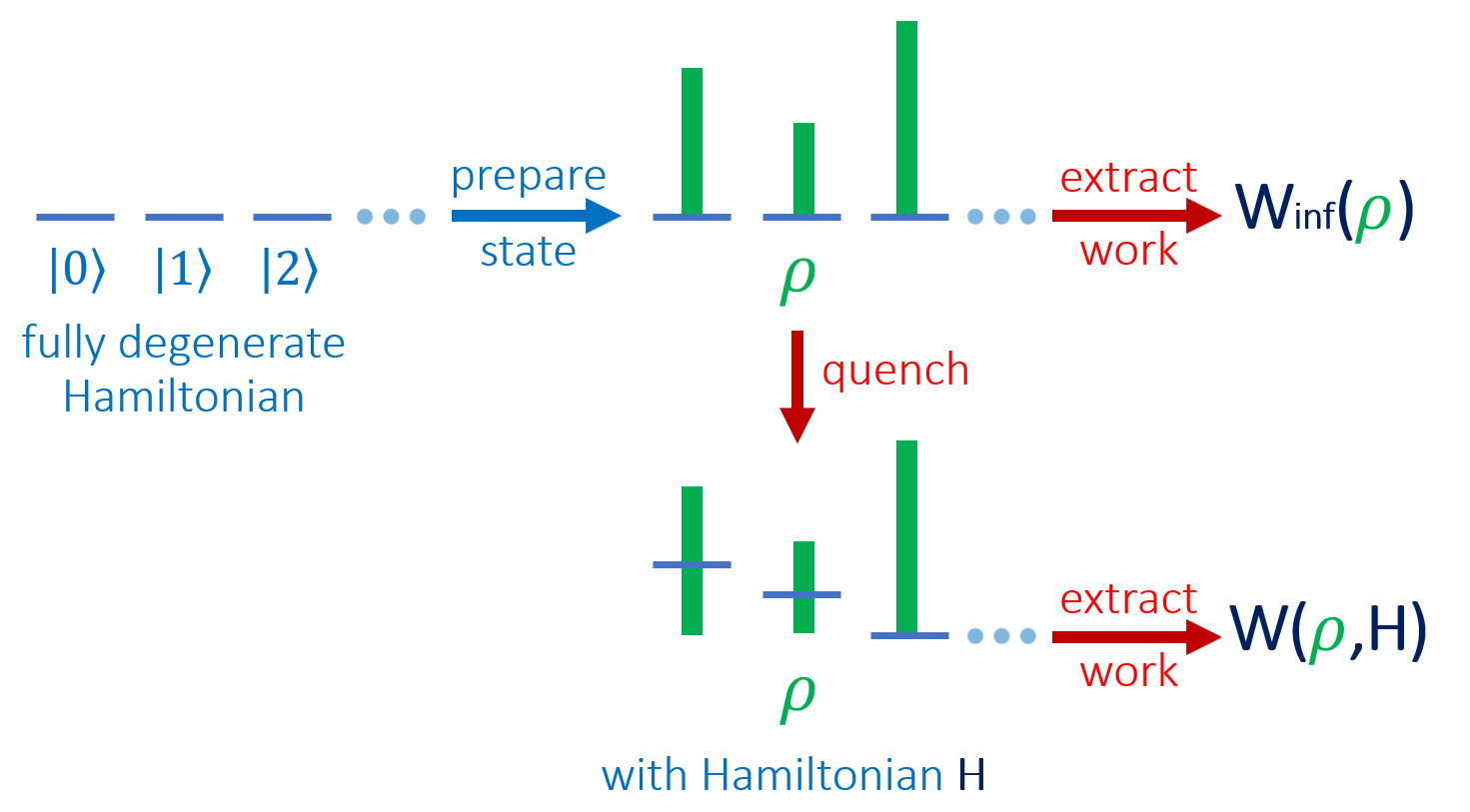}}
\caption{
{\bf \CY{Work extraction task that describes energy storage enhancement}.}
Each round starts with \CY{the same initial state $\rho$ subject to} a fully degenerate initial Hamiltonian.
\CY{If we extract work with this setting, we obtain $W_{\rm inf}(\rho)$ [Eq.~\eqref{Eq:Winf def}]. Alternatively, if we {\em quench} the Hamiltonian into a new one, $H$, and then perform work extraction, we obtain $W(\rho,H)$ [Eq.~\eqref{Eq:W def}].
The energy storage change} $\Wdiff(\rho,H)$ \CY{is} the difference between these two work values.
}
\label{Fig}
\end{figure}

{\bf\em Completely characterising state conversion by work extraction---}Here, by using $\Delta$, we can fully characterise conversion problems for general resource theories.
To see this, for a given state $\rho$ and Hamiltonian $H$, define
\begin{align}\label{Eq:Def Delta_OR}
\Delta_{\freeop}(\rho,H)\coloneqq\max_{\mathcal{N}\in\freeop}\Delta[\mathcal{N}(\rho),H],
\end{align}
which is the highest value of $\Delta$ that can be reached by applying allowed operations to the state $\rho$.
We thus call $\Delta_{\freeop}(\rho,H)$ the {\em $\freeop$-assisted energy storage enhancement} of $\rho$ with $H$.
Then, we have the following result, serving as a thermodynamic, complete characterisation of state conversion.
\begin{result}\label{Result:conversion}
Let \mbox{$0\le\epsilon<\delta$} be fixed energy scales.
Then \mbox{$\rho\stackrel{\freeop}{\longrightarrow}\sigma$} if and only if
\begin{align}\label{Eq:Result:conversion}
\Delta_{\freeop}(\rho,H)\ge\Delta_{\freeop}(\sigma,H)\quad\forall\,\epsilon\id\le H\le\delta\id.
\end{align}
\end{result}
We detail the proof in Appendix B.
Theorem~\ref{Result:conversion} gives the first thermodynamic characterisation, in a necessary and sufficient way, of general one-shot state conversion, i.e., it assumes access to a single copy of the state.
Here, $\epsilon$ and $\delta$ are freely chosen energy scales to fit the relevant physical settings. 
Hence, a complete characterisation can already be achieved by focusing on Hamiltonians with limited energy scales.
It is interesting to note that Theorem~\ref{Result:conversion} shows that the quantity $W(\rho,H)$ [Eq.~\eqref{Eq:W def}], which characterizes the \emph{asymptotically} extractable work~\cite{Skrzypczyk2014NC,Brandao2013PRL,Footnote1}, also completely characterises {\em one-shot} state conversion properties. We also note that although infinite conditions are necessary for general complete resource theories~\cite{Datta2023PRL}, {\em finite} conditions with practically feasible Hamiltonians may suffice for specific resources, as we shall demonstrate later.  

Interestingly, Eq.~\eqref{Eq:Result:conversion} takes a form that is similar to the so-called majorisation (see, e.g., Ref.~\cite{Purity-review}), as it implies that the state $\rho$ ``outperforms'' (or at least equals) $\sigma$ in terms of the figure-of-merit $\Delta_{\freeop}$ for every possible Hamiltonian in the given range.
We will see in the next section that, when applied to the appropriate setting, Theorem~\ref{Result:conversion} indeed reproduces the actual majorisation condition.

{\bf\em Implication for informational thermodynamics and entanglement---}
We now apply our result to the resource theory of {\em informational non-equilibrium} (see Ref.~\cite{Purity-review} for review).
For a $d$-dimensional system, the free set is \mbox{$\freeset=\{\id/d\}$,} corresponding to the thermal state \CYthree{[see Eq.~\eqref{Eq:thermal state}]} of a fully degenerate Hamiltonian.
In this case, with no energy differences, the thermodynamic contribution stems solely from the state’s information content.
The allowed operations $\freeop$ are all energy-conserving closed system dynamics and their classical mixtures, i.e., all mixed unitary channels. 
We write ``$\rho\stackrel{\rm MU}{\longrightarrow}\sigma$'' if there is one such channel converting $\rho$ into $\sigma$.
As we have a fixed system dimension, conversions under mixed unitary channels are equivalent to conversions under unital channels, denoted by ``$\rho\stackrel{\rm unital}{\longrightarrow}\sigma$'' (see Lemma 10 in Ref.~\cite{Purity-review}).
These state conversions can be characterised by the so-called majorisation. 
Formally, for a given state $\rho$, let $\{\rho_i^\downarrow\}_{i=0}^{d-1}$ be its eigenvalues in the \CYthree{non-increasing order,} i.e., $\rho_i^\downarrow\ge\rho_{i+1}^\downarrow$ $\forall\,i=0,...,d-2$.
Then, $\rho$ is said to {\em majorise} $\sigma$, denoted by $\rho\succ\sigma$, if (see, e.g., Ref.~\cite{Purity-review})
\begin{align}\label{Eq:majorisation}
\sum_{i=0}^k\rho_i^\downarrow\ge\sum_{i=0}^k\sigma_i^\downarrow\quad\forall\,k=0,...,d-2.
\end{align}
As detailed in Ref.~\cite{Purity-review}, by applying Hardy-Littlewood-P\'olya theorem~\cite{Hardy-Littlewood-Pólya} (see also Ref.~\footnote{\CY{To see this, it suffices to combine Lemmas 6, 10, 16 and Definition 15 in Ref.~\cite{Purity-review}.}}), we have the well-known result (here, ``iff'' means ``if and only if''): 
\begin{align}\label{Eq:Q2ndLaws info}
\rho\stackrel{\rm unital}{\longrightarrow}\sigma\quad\text{\em iff}\quad\rho\stackrel{\rm MU}{\longrightarrow}\sigma\quad\text{\em iff}\quad\rho\succ\sigma.
\end{align}
Namely, state conversions under allowed operations of informational non-equilibrium are fully captured by $d-1$ inequalities through majorisation Eq.~\eqref{Eq:majorisation}.
Here, as an application, we will demonstrate that Theorem~\ref{Result:conversion} reproduces the above result in this special case. Notably, \CYthree{as Theorem~\ref{Result:conversion}'s proof relies on Sion's minimax theorem~\cite{Sion1958}, rather than the Hardy-Littlewood-P\'olya theorem, we thus uncover a novel proof for this well-known result.}

Let $\{\ket{i}\}_{i=0}^{d-1}$ be a fixed basis and \mbox{$\rho^\downarrow\coloneqq\sum_{i=0}^{d-1}\rho_i^\downarrow\proj{i}$,} which is a ``reordered'' version of the state $\rho$ in this basis.
For $k=0,...,d-2$, consider Hamiltonians \mbox{$H_k\coloneqq\omega\sum_{i=0}^{k}\proj{i}$}, where $\omega>0$ is a fixed energy scale. 
Each $H_k$ has two different energy levels (i.e., $0$ and $\omega>0$) and $(d-k-1)$-fold ground state degeneracy.
Then, applying Theorem~\ref{Result:conversion} gives:
\begin{result}\label{coro}
For two $d$-dimensional states $\rho$ and $\sigma$, we have $\rho\stackrel{\rm unital}{\longrightarrow}\sigma$ if and only if $\rho\stackrel{\rm MU}{\longrightarrow}\sigma$ if and only if
\begin{align}\label{Eq:Info non-equilibrium result}
\Delta\left(\rho^\downarrow,H_k\right)\ge\Delta\left(\sigma^\downarrow,H_k\right)\quad\forall\,k=0,...,d-2.
\end{align}
Moreover, the above statement implies Eq.~\eqref{Eq:Q2ndLaws info}.
\end{result}
See Appendix C for proof.
This illustrates Theorem~\ref{Result:conversion}'s applicability to specific resources. 
In fact, Theorem~\ref{coro} goes beyond the existing notion of Eq.~\eqref{Eq:Q2ndLaws info} by linking state conversions to a {\em work-extraction task} with a clear thermodynamic interpretation. The involved Hamiltonians describe simple systems whose energy scales can be adjusted arbitrarily to fit practical settings \CYnew{(see also Ref.~\cite{Navascues2014PRL}).}
Since $\Delta$ can be obtained by two work extraction experiments (one for $W_{\rm inf}$; one for $W$), Theorem~\ref{coro} suggests that, in a $d$-dimensional system,
\begin{center}
{\em $2(d-1)$ work extraction experiments with simple Hamiltonians completely characterise state conversions for informational thermodynamics.}
\end{center}
Hence, our results not only generalise the well-known result for unital channel conversions [Eq.~\eqref{Eq:Q2ndLaws info}], but also provide an experimentally testable form for it via work extraction.

Next, we study the implications of Theorem~\ref{coro} for entanglement theory.
Suppose $\ket{\psi},\ket{\phi}$ are bipartite pure states with equal local dimension $d$, and $\psi,\phi$ are their reduced states in the first sub-system. 
The well-known Nielsen's theorem~\cite{QIC-book,Nielsen1999PRL} states that $\ket{\psi}\stackrel{\rm LOCC}{\longrightarrow}\ket{\phi}$ if and only if $\psi\prec\phi$, where \mbox{``$\ket{\psi}\stackrel{\rm LOCC}{\longrightarrow}\ket{\phi}$''} means that one can convert $\ket{\psi}$ into $\ket{\phi}$ by LOCC.
Then, Theorem~\ref{coro} and Nielsen's theorem jointly imply that 
$\ket{\psi}\stackrel{\rm LOCC}{\longrightarrow}\ket{\phi}$ if and only if
\begin{align}
\Delta\left(\psi^\downarrow,H_k\right)\le\Delta\left(\phi^\downarrow,H_k\right)\quad\forall\,k=0,...,d-2.
\end{align}
Hence, Theorem~\ref{coro} uncovers a {\em thermodynamic form} of Nielsen's theorem: $2(d-1)$ different work extraction experiments {\em performed in the local system} are enough to fully characterise pure state conversions under LOCC.
Further investigation into our findings' applications to entanglement (e.g., the connection with Refs.~\cite{Alimuddin2020PRE,Alimuddin2019PRA}) and other specific quantum resources is beyond the scope of this work and is left for future research.

{\bf\em Witnessing and quantifying quantum resources by work extraction---}We are now in the position to answer the question about the existence of a thermodynamic URCC.
From Theorem~\ref{Result:conversion} we can derive the following result:
\begin{result}\label{coro:state}
Let \mbox{$0\le\epsilon<\delta$} be fixed energy scales.
Then \mbox{$\rho\notin\freeset$} if and only if there is a Hamiltonian $\epsilon\id\le H\le\delta\id$ such that
\begin{align}
\Wdiff\left(\rho,H\right)>\max_{\eta\in\freeset}\Wdiff\left(\eta,H\right).
\end{align}
\end{result}
See Appendix D for the proof.
Hence, the energy storage enhancement $\Delta$ alone can certify state resources, and every resourceful state can provide an advantage over free states in enhancing energy storage under a specific quenching process.
Theorem~\ref{coro:state} thus establishes the first protocol to universally certify general quantum resources via energy extraction, and thus a thermodynamic URCC.

In fact, apart from certifying quantum resources, we can further {\em quantify} them by work extraction via $\Delta$.
To do so, we introduce the following measure for some fixed energy scales $0\le\epsilon<\delta$ (whose dependence will be kept implicit):
\begin{align}
\mathcal{M}_{\freeop}(\rho)\coloneqq
\max_{\epsilon\id\le H\le\delta\id}\left[\Delta_{\freeop}(\rho,H) - \max_{\eta\in\freeset}\Delta_{\freeop}(\eta,H)\right].
\end{align}
Then, we have the following result:
\begin{result}\label{result:quantification R}
For all energy scales $0\le\epsilon<\delta$, the measure $\mathcal{M}_{\freeop}$ is a resource quantifier of $R$; namely, it satisfies
\begin{enumerate}
\item$\mathcal{M}_{\freeop}(\rho)\ge0$ and $\mathcal{M}_{\freeop}(\rho)=0$ if and only if $\rho\in\freeset$.
\item$\mathcal{M}_{\freeop}[\mathcal{N}(\rho)]\le\mathcal{M}_{\freeop}(\rho)$ $\forall\,\mathcal{N}\in\freeop$.
\end{enumerate}
\end{result}
See Appendix E for its proof.
We thus uncover the first thermodynamic, energetic way to {\em quantify} general state resources, which is a strictly stronger notion than certification.

{\bf\em Mean energy as a witness of quantum resources.---}As a corollary of our approach, one can observe that Theorems~\ref{Result:conversion},~\ref{coro:state}, and~\ref{result:quantification R} also hold by replacing $\Delta(\rho,H)$ with the mean energy ${\rm tr}(\rho H)$. 
More specifically, Theorem~\ref{Result:conversion}'s proof can be adapted to the mean energy by starting from Eq.~\eqref{Eq: for mean energy proof} in Appendix B (which thus also implies Theorems~\ref{coro:state} and~\ref{result:quantification R} for the mean energy).
This means that a resourceful state can attain a higher mean energy (for some Hamiltonian) than any free state.
However, a state's mean energy may not directly equal its extractable work [which is more related to free energy; see also Eq.~\eqref{Eq:WgapRelation} in Appendix A].
Consequently, attaining a higher mean energy is still insufficient to imply an operational advantage in work extraction. This is why we presented our main results in terms of $\Delta$, which provides a clear operational picture.

{\bf\em Conclusions---}We have introduced a thermodynamic work extraction task that fully characterises state conversions in arbitrary resource theories (Theorem~\ref{Result:conversion}).
This framework yields the first thermodynamic universal certification class of quantum resources (Theorem~\ref{coro:state}) and provides a quantification of general resources based on work extraction (Theorem~\ref{result:quantification R}).
Our approach further suggests that mean energy can also characterise state conversions and hence witness quantum resources. 
Consequently, together with Ref.~\cite{deOliveiraJunior2025PRL}, we conclude that thermodynamic energies---{\em heat flow}, {\em mean energy}, and {\em extractable work}---all can witness general quantum properties.
This naturally raises the open question: {\em Can heat flow also provide a complete set of monotones for general quantum resources?} 

In the resource theory of informational non-equilibrium~\cite{Purity-review}, our result recovers the majorisation conditions for unital channel conversions and links them to work extraction (Theorem~\ref{coro}), while in entanglement theory it yields a thermodynamic interpretation of Nielsen’s theorem.
We conjecture that the same framework also reproduces thermo-majorisation in thermodynamics (see, e.g., Refs.\cite{Brandão2015, Ćwikliński2015PRL,Gour2018NC,Theurer2023NJP,Gour2022PRXQ,Lostaglio2019}), a question we leave for future work. 

Finally, for any resourceful state $\rho\notin\freeset$, by Theorem~\ref{coro:state}, there is a Hamiltonian $H>0$ achieving \mbox{$\Wdiff\left(\rho,H\right)>\max_{\eta\in\freeset}\Wdiff\left(\eta,H\right)$}.
By viewing the quenching as a ``charging process,'' Theorem~\ref{coro:state} means that under a {\em fixed} charging process with some final Hamiltonian $H$, a system initially in $\rho$ can be charged {\em strictly better} than one initially in any free state $\eta\in\freeset$. Namely, systems initially prepared in resourceful states can achieve larger energy storage enhancements under a fixed charging process. This thus raises the question of whether quantum batteries can be employed to characterize general state conversions via the standard ergotropy framework~\cite{Allahverdyan2004EPL,Tirone2024SciPostPhys,Yang2023PRL}.
Extending our methods to ergotropic approaches, such as quantum battery capacity~\cite{Yang2023PRL}, and broader notions of work extraction~\cite{Perarnau-LlobetPRX2015,Vinjanampathy2016CP,Ciampini2017npjQI,Francica2017npjQI,Andolina2019PRL,Monsel2020PRL,Opatrny2021PRL,Yang2023PRL} lies beyond the present scope and will be pursued in future work.

{\bf\em Acknowledgements---}We thank Antonio Ac\'in, \CY{Pharnam Bakhshinezhad,} Shin-Liang Chen, \CY{Shao-Hua Hu,} Huan-Yu Ku, \CY{Patryk Lipka-Bartosik,} Matteo Lostaglio, Paul Skrzypczyk, \CY{Alexander Streltsov}, and Benjamin Stratton for fruitful discussions and comments.
C.-Y.~H. \CY{acknowledges support from} the Royal Society through Enhanced Research Expenses (\CY{NFQI}), the ERC Advanced Grant (FLQuant), \CY{and the Leverhulme Trust Early Career Fellowship (``Quantum complementarity: a novel resource for quantum science and technologies'' with Grant No.~ECF-2024-310).} This work is supported by the project PID2023-152724NA-I00, with funding from MCIU/AEI/10.13039/501100011033 and FSE+, by the project CNS2024-154818 with funding by MICIU/AEI/10.13039/501100011033, by the project RYC2021-031094-I, with funding from MCIN/AEI/10.13039/501100011033 and the European Union ‘NextGenerationEU’ PRTR fund, by the project CIPROM/2022/66 with funding by the Generalitat Valenciana, and by the Ministry of Economic Affairs and Digital Transformation of the Spanish Government through the QUANTUM ENIA Project call—QUANTUM SPAIN Project, by the European Union through the Recovery, Transformation and Resilience Plan—NextGenerationEU within the framework of the Digital Spain 2026 Agenda, and by the CSIC Interdisciplinary Thematic Platform (PTI+) on Quantum Technologies (PTI-QTEP+). This work is supported through the project CEX2023-001292-S funded by MCIU/AEI.

\bibliography{Ref.bib}

\begin{thebibliography}{111}%
\makeatletter
\providecommand \@ifxundefined [1]{%
 \@ifx{#1\undefined}
}%
\providecommand \@ifnum [1]{%
 \ifnum #1\expandafter \@firstoftwo
 \else \expandafter \@secondoftwo
 \fi
}%
\providecommand \@ifx [1]{%
 \ifx #1\expandafter \@firstoftwo
 \else \expandafter \@secondoftwo
 \fi
}%
\providecommand \natexlab [1]{#1}%
\providecommand \enquote  [1]{``#1''}%
\providecommand \bibnamefont  [1]{#1}%
\providecommand \bibfnamefont [1]{#1}%
\providecommand \citenamefont [1]{#1}%
\providecommand \href@noop [0]{\@secondoftwo}%
\providecommand \href [0]{\begingroup \@sanitize@url \@href}%
\providecommand \@href[1]{\@@startlink{#1}\@@href}%
\providecommand \@@href[1]{\endgroup#1\@@endlink}%
\providecommand \@sanitize@url [0]{\catcode `\\12\catcode `\$12\catcode `\&12\catcode `\#12\catcode `\^12\catcode `\_12\catcode `\%12\relax}%
\providecommand \@@startlink[1]{}%
\providecommand \@@endlink[0]{}%
\providecommand \url  [0]{\begingroup\@sanitize@url \@url }%
\providecommand \@url [1]{\endgroup\@href {#1}{\urlprefix }}%
\providecommand \urlprefix  [0]{URL }%
\providecommand \Eprint [0]{\href }%
\providecommand \doibase [0]{https://doi.org/}%
\providecommand \selectlanguage [0]{\@gobble}%
\providecommand \bibinfo  [0]{\@secondoftwo}%
\providecommand \bibfield  [0]{\@secondoftwo}%
\providecommand \translation [1]{[#1]}%
\providecommand \BibitemOpen [0]{}%
\providecommand \bibitemStop [0]{}%
\providecommand \bibitemNoStop [0]{.\EOS\space}%
\providecommand \EOS [0]{\spacefactor3000\relax}%
\providecommand \BibitemShut  [1]{\csname bibitem#1\endcsname}%
\let\auto@bib@innerbib\@empty
\bibitem [{\citenamefont {Kuroiwa}\ \emph {et~al.}(2024{\natexlab{a}})\citenamefont {Kuroiwa}, \citenamefont {Takagi}, \citenamefont {Adesso},\ and\ \citenamefont {Yamasaki}}]{Kuroiwa2024PRL}%
  \BibitemOpen
  \bibfield  {author} {\bibinfo {author} {\bibfnamefont {K.}~\bibnamefont {Kuroiwa}}, \bibinfo {author} {\bibfnamefont {R.}~\bibnamefont {Takagi}}, \bibinfo {author} {\bibfnamefont {G.}~\bibnamefont {Adesso}},\ and\ \bibinfo {author} {\bibfnamefont {H.}~\bibnamefont {Yamasaki}},\ }\bibfield  {title} {\bibinfo {title} {Every quantum helps: Operational advantage of quantum resources beyond convexity},\ }\href {https://doi.org/10.1103/PhysRevLett.132.150201} {\bibfield  {journal} {\bibinfo  {journal} {Phys. Rev. Lett.}\ }\textbf {\bibinfo {volume} {132}},\ \bibinfo {pages} {150201} (\bibinfo {year} {2024}{\natexlab{a}})}\BibitemShut {NoStop}%
\bibitem [{\citenamefont {Kuroiwa}\ \emph {et~al.}(2024{\natexlab{b}})\citenamefont {Kuroiwa}, \citenamefont {Takagi}, \citenamefont {Adesso},\ and\ \citenamefont {Yamasaki}}]{Kuroiwa2024PRA}%
  \BibitemOpen
  \bibfield  {author} {\bibinfo {author} {\bibfnamefont {K.}~\bibnamefont {Kuroiwa}}, \bibinfo {author} {\bibfnamefont {R.}~\bibnamefont {Takagi}}, \bibinfo {author} {\bibfnamefont {G.}~\bibnamefont {Adesso}},\ and\ \bibinfo {author} {\bibfnamefont {H.}~\bibnamefont {Yamasaki}},\ }\bibfield  {title} {\bibinfo {title} {Robustness- and weight-based resource measures without convexity restriction: Multicopy witness and operational advantage in static and dynamical quantum resource theories},\ }\href {https://doi.org/10.1103/PhysRevA.109.042403} {\bibfield  {journal} {\bibinfo  {journal} {Phys. Rev. A}\ }\textbf {\bibinfo {volume} {109}},\ \bibinfo {pages} {042403} (\bibinfo {year} {2024}{\natexlab{b}})}\BibitemShut {NoStop}%
\bibitem [{\citenamefont {Meier}\ and\ \citenamefont {Yamasaki}(2025)}]{Meier2025PRXEnergy}%
  \BibitemOpen
  \bibfield  {author} {\bibinfo {author} {\bibfnamefont {F.}~\bibnamefont {Meier}}\ and\ \bibinfo {author} {\bibfnamefont {H.}~\bibnamefont {Yamasaki}},\ }\bibfield  {title} {\bibinfo {title} {Energy-consumption advantage of quantum computation},\ }\href {https://doi.org/10.1103/PRXEnergy.4.023008} {\bibfield  {journal} {\bibinfo  {journal} {PRX Energy}\ }\textbf {\bibinfo {volume} {4}},\ \bibinfo {pages} {023008} (\bibinfo {year} {2025})}\BibitemShut {NoStop}%
\bibitem [{\citenamefont {Bennett}\ \emph {et~al.}(1993)\citenamefont {Bennett}, \citenamefont {Brassard}, \citenamefont {Cr\'epeau}, \citenamefont {Jozsa}, \citenamefont {Peres},\ and\ \citenamefont {Wootters}}]{Bennett93}%
  \BibitemOpen
  \bibfield  {author} {\bibinfo {author} {\bibfnamefont {C.~H.}\ \bibnamefont {Bennett}}, \bibinfo {author} {\bibfnamefont {G.}~\bibnamefont {Brassard}}, \bibinfo {author} {\bibfnamefont {C.}~\bibnamefont {Cr\'epeau}}, \bibinfo {author} {\bibfnamefont {R.}~\bibnamefont {Jozsa}}, \bibinfo {author} {\bibfnamefont {A.}~\bibnamefont {Peres}},\ and\ \bibinfo {author} {\bibfnamefont {W.~K.}\ \bibnamefont {Wootters}},\ }\bibfield  {title} {\bibinfo {title} {Teleporting an unknown quantum state via dual classical and {E}instein-{P}odolsky-{R}osen channels},\ }\href {https://doi.org/10.1103/PhysRevLett.70.1895} {\bibfield  {journal} {\bibinfo  {journal} {Phys. Rev. Lett.}\ }\textbf {\bibinfo {volume} {70}},\ \bibinfo {pages} {1895} (\bibinfo {year} {1993})}\BibitemShut {NoStop}%
\bibitem [{\citenamefont {Bennett}\ and\ \citenamefont {Wiesner}(1992)}]{Bennett92}%
  \BibitemOpen
  \bibfield  {author} {\bibinfo {author} {\bibfnamefont {C.~H.}\ \bibnamefont {Bennett}}\ and\ \bibinfo {author} {\bibfnamefont {S.~J.}\ \bibnamefont {Wiesner}},\ }\bibfield  {title} {\bibinfo {title} {Communication via one- and two-particle operators on {E}instein-{P}odolsky-{R}osen states},\ }\href {https://doi.org/10.1103/PhysRevLett.69.2881} {\bibfield  {journal} {\bibinfo  {journal} {Phys. Rev. Lett.}\ }\textbf {\bibinfo {volume} {69}},\ \bibinfo {pages} {2881} (\bibinfo {year} {1992})}\BibitemShut {NoStop}%
\bibitem [{\citenamefont {Pezz\`e}\ \emph {et~al.}(2018)\citenamefont {Pezz\`e}, \citenamefont {Smerzi}, \citenamefont {Oberthaler}, \citenamefont {Schmied},\ and\ \citenamefont {Treutlein}}]{RevModPhys.90.035005}%
  \BibitemOpen
  \bibfield  {author} {\bibinfo {author} {\bibfnamefont {L.}~\bibnamefont {Pezz\`e}}, \bibinfo {author} {\bibfnamefont {A.}~\bibnamefont {Smerzi}}, \bibinfo {author} {\bibfnamefont {M.~K.}\ \bibnamefont {Oberthaler}}, \bibinfo {author} {\bibfnamefont {R.}~\bibnamefont {Schmied}},\ and\ \bibinfo {author} {\bibfnamefont {P.}~\bibnamefont {Treutlein}},\ }\bibfield  {title} {\bibinfo {title} {Quantum metrology with nonclassical states of atomic ensembles},\ }\href {https://doi.org/10.1103/RevModPhys.90.035005} {\bibfield  {journal} {\bibinfo  {journal} {Rev. Mod. Phys.}\ }\textbf {\bibinfo {volume} {90}},\ \bibinfo {pages} {035005} (\bibinfo {year} {2018})}\BibitemShut {NoStop}%
\bibitem [{\citenamefont {Ren}\ \emph {et~al.}(2021)\citenamefont {Ren}, \citenamefont {Li}, \citenamefont {Smerzi},\ and\ \citenamefont {Gessner}}]{PhysRevLett.126.080502}%
  \BibitemOpen
  \bibfield  {author} {\bibinfo {author} {\bibfnamefont {Z.}~\bibnamefont {Ren}}, \bibinfo {author} {\bibfnamefont {W.}~\bibnamefont {Li}}, \bibinfo {author} {\bibfnamefont {A.}~\bibnamefont {Smerzi}},\ and\ \bibinfo {author} {\bibfnamefont {M.}~\bibnamefont {Gessner}},\ }\bibfield  {title} {\bibinfo {title} {Metrological detection of multipartite entanglement from young diagrams},\ }\href {https://doi.org/10.1103/PhysRevLett.126.080502} {\bibfield  {journal} {\bibinfo  {journal} {Phys. Rev. Lett.}\ }\textbf {\bibinfo {volume} {126}},\ \bibinfo {pages} {080502} (\bibinfo {year} {2021})}\BibitemShut {NoStop}%
\bibitem [{\citenamefont {Horodecki}\ \emph {et~al.}(2009)\citenamefont {Horodecki}, \citenamefont {Horodecki}, \citenamefont {Horodecki},\ and\ \citenamefont {Horodecki}}]{HorodeckiRMP}%
  \BibitemOpen
  \bibfield  {author} {\bibinfo {author} {\bibfnamefont {R.}~\bibnamefont {Horodecki}}, \bibinfo {author} {\bibfnamefont {P.}~\bibnamefont {Horodecki}}, \bibinfo {author} {\bibfnamefont {M.}~\bibnamefont {Horodecki}},\ and\ \bibinfo {author} {\bibfnamefont {K.}~\bibnamefont {Horodecki}},\ }\bibfield  {title} {\bibinfo {title} {Quantum entanglement},\ }\href {https://doi.org/10.1103/RevModPhys.81.865} {\bibfield  {journal} {\bibinfo  {journal} {Rev. Mod. Phys.}\ }\textbf {\bibinfo {volume} {81}},\ \bibinfo {pages} {865} (\bibinfo {year} {2009})}\BibitemShut {NoStop}%
\bibitem [{\citenamefont {Brunner}\ \emph {et~al.}(2014)\citenamefont {Brunner}, \citenamefont {Cavalcanti}, \citenamefont {Pironio}, \citenamefont {Scarani},\ and\ \citenamefont {Wehner}}]{Brunner2014RMP}%
  \BibitemOpen
  \bibfield  {author} {\bibinfo {author} {\bibfnamefont {N.}~\bibnamefont {Brunner}}, \bibinfo {author} {\bibfnamefont {D.}~\bibnamefont {Cavalcanti}}, \bibinfo {author} {\bibfnamefont {S.}~\bibnamefont {Pironio}}, \bibinfo {author} {\bibfnamefont {V.}~\bibnamefont {Scarani}},\ and\ \bibinfo {author} {\bibfnamefont {S.}~\bibnamefont {Wehner}},\ }\bibfield  {title} {\bibinfo {title} {Bell nonlocality},\ }\href {https://doi.org/10.1103/RevModPhys.86.419} {\bibfield  {journal} {\bibinfo  {journal} {Rev. Mod. Phys.}\ }\textbf {\bibinfo {volume} {86}},\ \bibinfo {pages} {419} (\bibinfo {year} {2014})}\BibitemShut {NoStop}%
\bibitem [{\citenamefont {Ac\'{\i}n}\ \emph {et~al.}(2007)\citenamefont {Ac\'{\i}n}, \citenamefont {Brunner}, \citenamefont {Gisin}, \citenamefont {Massar}, \citenamefont {Pironio},\ and\ \citenamefont {Scarani}}]{Acin2007PRL}%
  \BibitemOpen
  \bibfield  {author} {\bibinfo {author} {\bibfnamefont {A.}~\bibnamefont {Ac\'{\i}n}}, \bibinfo {author} {\bibfnamefont {N.}~\bibnamefont {Brunner}}, \bibinfo {author} {\bibfnamefont {N.}~\bibnamefont {Gisin}}, \bibinfo {author} {\bibfnamefont {S.}~\bibnamefont {Massar}}, \bibinfo {author} {\bibfnamefont {S.}~\bibnamefont {Pironio}},\ and\ \bibinfo {author} {\bibfnamefont {V.}~\bibnamefont {Scarani}},\ }\bibfield  {title} {\bibinfo {title} {Device-independent security of quantum cryptography against collective attacks},\ }\href {https://doi.org/10.1103/PhysRevLett.98.230501} {\bibfield  {journal} {\bibinfo  {journal} {Phys. Rev. Lett.}\ }\textbf {\bibinfo {volume} {98}},\ \bibinfo {pages} {230501} (\bibinfo {year} {2007})}\BibitemShut {NoStop}%
\bibitem [{\citenamefont {Uola}\ \emph {et~al.}(2020{\natexlab{a}})\citenamefont {Uola}, \citenamefont {Costa}, \citenamefont {Nguyen},\ and\ \citenamefont {G\"uhne}}]{UolaRMP2020}%
  \BibitemOpen
  \bibfield  {author} {\bibinfo {author} {\bibfnamefont {R.}~\bibnamefont {Uola}}, \bibinfo {author} {\bibfnamefont {A.~C.~S.}\ \bibnamefont {Costa}}, \bibinfo {author} {\bibfnamefont {H.~C.}\ \bibnamefont {Nguyen}},\ and\ \bibinfo {author} {\bibfnamefont {O.}~\bibnamefont {G\"uhne}},\ }\bibfield  {title} {\bibinfo {title} {Quantum steering},\ }\href {https://doi.org/10.1103/RevModPhys.92.015001} {\bibfield  {journal} {\bibinfo  {journal} {Rev. Mod. Phys.}\ }\textbf {\bibinfo {volume} {92}},\ \bibinfo {pages} {015001} (\bibinfo {year} {2020}{\natexlab{a}})}\BibitemShut {NoStop}%
\bibitem [{\citenamefont {Cavalcanti}\ and\ \citenamefont {Skrzypczyk}(2016)}]{Cavalcanti2016}%
  \BibitemOpen
  \bibfield  {author} {\bibinfo {author} {\bibfnamefont {D.}~\bibnamefont {Cavalcanti}}\ and\ \bibinfo {author} {\bibfnamefont {P.}~\bibnamefont {Skrzypczyk}},\ }\bibfield  {title} {\bibinfo {title} {Quantum steering: a review with focus on semidefinite programming},\ }\href {https://doi.org/10.1088/1361-6633/80/2/024001} {\bibfield  {journal} {\bibinfo  {journal} {Rep. Prog. Phys.}\ }\textbf {\bibinfo {volume} {80}},\ \bibinfo {pages} {024001} (\bibinfo {year} {2016})}\BibitemShut {NoStop}%
\bibitem [{\citenamefont {Branciard}\ \emph {et~al.}(2012)\citenamefont {Branciard}, \citenamefont {Cavalcanti}, \citenamefont {Walborn}, \citenamefont {Scarani},\ and\ \citenamefont {Wiseman}}]{Branciard2012PRA}%
  \BibitemOpen
  \bibfield  {author} {\bibinfo {author} {\bibfnamefont {C.}~\bibnamefont {Branciard}}, \bibinfo {author} {\bibfnamefont {E.~G.}\ \bibnamefont {Cavalcanti}}, \bibinfo {author} {\bibfnamefont {S.~P.}\ \bibnamefont {Walborn}}, \bibinfo {author} {\bibfnamefont {V.}~\bibnamefont {Scarani}},\ and\ \bibinfo {author} {\bibfnamefont {H.~M.}\ \bibnamefont {Wiseman}},\ }\bibfield  {title} {\bibinfo {title} {One-sided device-independent quantum key distribution: Security, feasibility, and the connection with steering},\ }\href {https://doi.org/10.1103/PhysRevA.85.010301} {\bibfield  {journal} {\bibinfo  {journal} {Phys. Rev. A}\ }\textbf {\bibinfo {volume} {85}},\ \bibinfo {pages} {010301(R)} (\bibinfo {year} {2012})}\BibitemShut {NoStop}%
\bibitem [{\citenamefont {G\"uhne}\ \emph {et~al.}(2023)\citenamefont {G\"uhne}, \citenamefont {Haapasalo}, \citenamefont {Kraft}, \citenamefont {Pellonp\"a\"a},\ and\ \citenamefont {Uola}}]{Otfried2021Rev}%
  \BibitemOpen
  \bibfield  {author} {\bibinfo {author} {\bibfnamefont {O.}~\bibnamefont {G\"uhne}}, \bibinfo {author} {\bibfnamefont {E.}~\bibnamefont {Haapasalo}}, \bibinfo {author} {\bibfnamefont {T.}~\bibnamefont {Kraft}}, \bibinfo {author} {\bibfnamefont {J.-P.}\ \bibnamefont {Pellonp\"a\"a}},\ and\ \bibinfo {author} {\bibfnamefont {R.}~\bibnamefont {Uola}},\ }\bibfield  {title} {\bibinfo {title} {Colloquium: Incompatible measurements in quantum information science},\ }\href {https://doi.org/10.1103/RevModPhys.95.011003} {\bibfield  {journal} {\bibinfo  {journal} {Rev. Mod. Phys.}\ }\textbf {\bibinfo {volume} {95}},\ \bibinfo {pages} {011003} (\bibinfo {year} {2023})}\BibitemShut {NoStop}%
\bibitem [{\citenamefont {Hsieh}\ \emph {et~al.}(2023)\citenamefont {Hsieh}, \citenamefont {Uola},\ and\ \citenamefont {Skrzypczyk}}]{Hsieh2023}%
  \BibitemOpen
  \bibfield  {author} {\bibinfo {author} {\bibfnamefont {C.-Y.}\ \bibnamefont {Hsieh}}, \bibinfo {author} {\bibfnamefont {R.}~\bibnamefont {Uola}},\ and\ \bibinfo {author} {\bibfnamefont {P.}~\bibnamefont {Skrzypczyk}},\ }\href@noop {} {\bibinfo {title} {Quantum complementarity: A novel resource for unambiguous exclusion and encryption}} (\bibinfo {year} {2023}),\ \Eprint {https://arxiv.org/abs/2309.11968} {arXiv:2309.11968} \BibitemShut {NoStop}%
\bibitem [{\citenamefont {Rosset}\ \emph {et~al.}(2018)\citenamefont {Rosset}, \citenamefont {Buscemi},\ and\ \citenamefont {Liang}}]{Rosset2018PRX}%
  \BibitemOpen
  \bibfield  {author} {\bibinfo {author} {\bibfnamefont {D.}~\bibnamefont {Rosset}}, \bibinfo {author} {\bibfnamefont {F.}~\bibnamefont {Buscemi}},\ and\ \bibinfo {author} {\bibfnamefont {Y.-C.}\ \bibnamefont {Liang}},\ }\bibfield  {title} {\bibinfo {title} {Resource theory of quantum memories and their faithful verification with minimal assumptions},\ }\href {https://doi.org/10.1103/PhysRevX.8.021033} {\bibfield  {journal} {\bibinfo  {journal} {Phys. Rev. X}\ }\textbf {\bibinfo {volume} {8}},\ \bibinfo {pages} {021033} (\bibinfo {year} {2018})}\BibitemShut {NoStop}%
\bibitem [{\citenamefont {Yuan}\ \emph {et~al.}(2021)\citenamefont {Yuan}, \citenamefont {Liu}, \citenamefont {Zhao}, \citenamefont {Regula}, \citenamefont {Thompson},\ and\ \citenamefont {Gu}}]{Yuan2021npjQI}%
  \BibitemOpen
  \bibfield  {author} {\bibinfo {author} {\bibfnamefont {X.}~\bibnamefont {Yuan}}, \bibinfo {author} {\bibfnamefont {Y.}~\bibnamefont {Liu}}, \bibinfo {author} {\bibfnamefont {Q.}~\bibnamefont {Zhao}}, \bibinfo {author} {\bibfnamefont {B.}~\bibnamefont {Regula}}, \bibinfo {author} {\bibfnamefont {J.}~\bibnamefont {Thompson}},\ and\ \bibinfo {author} {\bibfnamefont {M.}~\bibnamefont {Gu}},\ }\bibfield  {title} {\bibinfo {title} {Universal and operational benchmarking of quantum memories},\ }\href {https://doi.org/10.1038/s41534-021-00444-9} {\bibfield  {journal} {\bibinfo  {journal} {npj Quantum Inf.}\ }\textbf {\bibinfo {volume} {7}},\ \bibinfo {pages} {108} (\bibinfo {year} {2021})}\BibitemShut {NoStop}%
\bibitem [{\citenamefont {Ku}\ \emph {et~al.}(2022)\citenamefont {Ku}, \citenamefont {Kadlec}, \citenamefont {\ifmmode~\check{C}\else \v{C}\fi{}ernoch}, \citenamefont {Quintino}, \citenamefont {Zhou}, \citenamefont {Lemr}, \citenamefont {Lambert}, \citenamefont {Miranowicz}, \citenamefont {Chen}, \citenamefont {Nori},\ and\ \citenamefont {Chen}}]{Ku2022PRXQ}%
  \BibitemOpen
  \bibfield  {author} {\bibinfo {author} {\bibfnamefont {H.-Y.}\ \bibnamefont {Ku}}, \bibinfo {author} {\bibfnamefont {J.}~\bibnamefont {Kadlec}}, \bibinfo {author} {\bibfnamefont {A.}~\bibnamefont {\ifmmode~\check{C}\else \v{C}\fi{}ernoch}}, \bibinfo {author} {\bibfnamefont {M.~T.}\ \bibnamefont {Quintino}}, \bibinfo {author} {\bibfnamefont {W.}~\bibnamefont {Zhou}}, \bibinfo {author} {\bibfnamefont {K.}~\bibnamefont {Lemr}}, \bibinfo {author} {\bibfnamefont {N.}~\bibnamefont {Lambert}}, \bibinfo {author} {\bibfnamefont {A.}~\bibnamefont {Miranowicz}}, \bibinfo {author} {\bibfnamefont {S.-L.}\ \bibnamefont {Chen}}, \bibinfo {author} {\bibfnamefont {F.}~\bibnamefont {Nori}},\ and\ \bibinfo {author} {\bibfnamefont {Y.-N.}\ \bibnamefont {Chen}},\ }\bibfield  {title} {\bibinfo {title} {Quantifying quantumness of channels without entanglement},\ }\href {https://doi.org/10.1103/PRXQuantum.3.020338} {\bibfield  {journal} {\bibinfo  {journal} {PRX Quantum}\ }\textbf {\bibinfo {volume} {3}},\ \bibinfo {pages}
  {020338} (\bibinfo {year} {2022})}\BibitemShut {NoStop}%
\bibitem [{\citenamefont {Vieira}\ \emph {et~al.}(2025)\citenamefont {Vieira}, \citenamefont {Ku},\ and\ \citenamefont {Budroni}}]{Vieira2024}%
  \BibitemOpen
  \bibfield  {author} {\bibinfo {author} {\bibfnamefont {L.~B.}\ \bibnamefont {Vieira}}, \bibinfo {author} {\bibfnamefont {H.-Y.}\ \bibnamefont {Ku}},\ and\ \bibinfo {author} {\bibfnamefont {C.}~\bibnamefont {Budroni}},\ }\bibfield  {title} {\bibinfo {title} {Entanglement-breaking channels are a quantum memory resource},\ }\href {https://doi.org/10.1103/r8lf-bb4p} {\bibfield  {journal} {\bibinfo  {journal} {Phys. Rev. Res.}\ }\textbf {\bibinfo {volume} {7}},\ \bibinfo {pages} {043281} (\bibinfo {year} {2025})}\BibitemShut {NoStop}%
\bibitem [{\citenamefont {Abiuso}(2023)}]{Abiuso2024}%
  \BibitemOpen
  \bibfield  {author} {\bibinfo {author} {\bibfnamefont {P.}~\bibnamefont {Abiuso}},\ }\bibfield  {title} {\bibinfo {title} {Verification of continuous-variable quantum memories},\ }\href {https://doi.org/10.1088/2058-9565/ad097c} {\bibfield  {journal} {\bibinfo  {journal} {Quantum Sci. Technol.}\ }\textbf {\bibinfo {volume} {9}},\ \bibinfo {pages} {01LT02} (\bibinfo {year} {2023})}\BibitemShut {NoStop}%
\bibitem [{\citenamefont {Narasimhachar}\ \emph {et~al.}(2019)\citenamefont {Narasimhachar}, \citenamefont {Thompson}, \citenamefont {Ma}, \citenamefont {Gour},\ and\ \citenamefont {Gu}}]{NarasimhacharPRL2019}%
  \BibitemOpen
  \bibfield  {author} {\bibinfo {author} {\bibfnamefont {V.}~\bibnamefont {Narasimhachar}}, \bibinfo {author} {\bibfnamefont {J.}~\bibnamefont {Thompson}}, \bibinfo {author} {\bibfnamefont {J.}~\bibnamefont {Ma}}, \bibinfo {author} {\bibfnamefont {G.}~\bibnamefont {Gour}},\ and\ \bibinfo {author} {\bibfnamefont {M.}~\bibnamefont {Gu}},\ }\bibfield  {title} {\bibinfo {title} {Quantifying memory capacity as a quantum thermodynamic resource},\ }\href {https://doi.org/10.1103/PhysRevLett.122.060601} {\bibfield  {journal} {\bibinfo  {journal} {Phys. Rev. Lett.}\ }\textbf {\bibinfo {volume} {122}},\ \bibinfo {pages} {060601} (\bibinfo {year} {2019})}\BibitemShut {NoStop}%
\bibitem [{\citenamefont {Hsieh}\ \emph {et~al.}(2025)\citenamefont {Hsieh}, \citenamefont {Stratton}, \citenamefont {Wu},\ and\ \citenamefont {Ku}}]{Hsieh2025PRA-3}%
  \BibitemOpen
  \bibfield  {author} {\bibinfo {author} {\bibfnamefont {C.-Y.}\ \bibnamefont {Hsieh}}, \bibinfo {author} {\bibfnamefont {B.}~\bibnamefont {Stratton}}, \bibinfo {author} {\bibfnamefont {C.-H.}\ \bibnamefont {Wu}},\ and\ \bibinfo {author} {\bibfnamefont {H.-Y.}\ \bibnamefont {Ku}},\ }\bibfield  {title} {\bibinfo {title} {Dynamical resource theory of incompatibility preservability},\ }\href {https://doi.org/10.1103/PhysRevA.111.022422} {\bibfield  {journal} {\bibinfo  {journal} {Phys. Rev. A}\ }\textbf {\bibinfo {volume} {111}},\ \bibinfo {pages} {022422} (\bibinfo {year} {2025})}\BibitemShut {NoStop}%
\bibitem [{\citenamefont {Takagi}\ \emph {et~al.}(2020)\citenamefont {Takagi}, \citenamefont {Wang},\ and\ \citenamefont {Hayashi}}]{Takagi2020PRL}%
  \BibitemOpen
  \bibfield  {author} {\bibinfo {author} {\bibfnamefont {R.}~\bibnamefont {Takagi}}, \bibinfo {author} {\bibfnamefont {K.}~\bibnamefont {Wang}},\ and\ \bibinfo {author} {\bibfnamefont {M.}~\bibnamefont {Hayashi}},\ }\bibfield  {title} {\bibinfo {title} {Application of the resource theory of channels to communication scenarios},\ }\href {https://doi.org/10.1103/PhysRevLett.124.120502} {\bibfield  {journal} {\bibinfo  {journal} {Phys. Rev. Lett.}\ }\textbf {\bibinfo {volume} {124}},\ \bibinfo {pages} {120502} (\bibinfo {year} {2020})}\BibitemShut {NoStop}%
\bibitem [{\citenamefont {Hsieh}(2021)}]{Hsieh2021PRXQ}%
  \BibitemOpen
  \bibfield  {author} {\bibinfo {author} {\bibfnamefont {C.-Y.}\ \bibnamefont {Hsieh}},\ }\bibfield  {title} {\bibinfo {title} {Communication, dynamical resource theory, and thermodynamics},\ }\href {https://doi.org/10.1103/PRXQuantum.2.020318} {\bibfield  {journal} {\bibinfo  {journal} {PRX Quantum}\ }\textbf {\bibinfo {volume} {2}},\ \bibinfo {pages} {020318} (\bibinfo {year} {2021})}\BibitemShut {NoStop}%
\bibitem [{\citenamefont {Hsieh}(2025{\natexlab{a}})}]{Hsieh2025PRL}%
  \BibitemOpen
  \bibfield  {author} {\bibinfo {author} {\bibfnamefont {C.-Y.}\ \bibnamefont {Hsieh}},\ }\bibfield  {title} {\bibinfo {title} {Dynamical {L}andauer principle: Quantifying information transmission by thermodynamics},\ }\href {https://doi.org/10.1103/PhysRevLett.134.050404} {\bibfield  {journal} {\bibinfo  {journal} {Phys. Rev. Lett.}\ }\textbf {\bibinfo {volume} {134}},\ \bibinfo {pages} {050404} (\bibinfo {year} {2025}{\natexlab{a}})}\BibitemShut {NoStop}%
\bibitem [{\citenamefont {Hsieh}(2025{\natexlab{b}})}]{Hsieh2025PRA}%
  \BibitemOpen
  \bibfield  {author} {\bibinfo {author} {\bibfnamefont {C.-Y.}\ \bibnamefont {Hsieh}},\ }\bibfield  {title} {\bibinfo {title} {Dynamical {L}andauer principle: Thermodynamic criteria of transmitting classical information},\ }\href {https://doi.org/10.1103/PhysRevA.111.022207} {\bibfield  {journal} {\bibinfo  {journal} {Phys. Rev. A}\ }\textbf {\bibinfo {volume} {111}},\ \bibinfo {pages} {022207} (\bibinfo {year} {2025}{\natexlab{b}})}\BibitemShut {NoStop}%
\bibitem [{\citenamefont {Hsieh}(2020)}]{Hsieh2020}%
  \BibitemOpen
  \bibfield  {author} {\bibinfo {author} {\bibfnamefont {C.-Y.}\ \bibnamefont {Hsieh}},\ }\bibfield  {title} {\bibinfo {title} {Resource preservability},\ }\href {https://doi.org/10.22331/q-2020-03-19-244} {\bibfield  {journal} {\bibinfo  {journal} {{Quantum}}\ }\textbf {\bibinfo {volume} {4}},\ \bibinfo {pages} {244} (\bibinfo {year} {2020})}\BibitemShut {NoStop}%
\bibitem [{\citenamefont {Liu}\ and\ \citenamefont {Winter}(2019)}]{Liu2019DRT}%
  \BibitemOpen
  \bibfield  {author} {\bibinfo {author} {\bibfnamefont {Z.-W.}\ \bibnamefont {Liu}}\ and\ \bibinfo {author} {\bibfnamefont {A.}~\bibnamefont {Winter}},\ }\href@noop {} {\bibinfo {title} {Resource theories of quantum channels and the universal role of resource erasure}} (\bibinfo {year} {2019}),\ \Eprint {https://arxiv.org/abs/1904.04201} {arXiv:1904.04201} \BibitemShut {NoStop}%
\bibitem [{\citenamefont {Liu}\ and\ \citenamefont {Yuan}(2020)}]{Liu2020PRR}%
  \BibitemOpen
  \bibfield  {author} {\bibinfo {author} {\bibfnamefont {Y.}~\bibnamefont {Liu}}\ and\ \bibinfo {author} {\bibfnamefont {X.}~\bibnamefont {Yuan}},\ }\bibfield  {title} {\bibinfo {title} {Operational resource theory of quantum channels},\ }\href {https://doi.org/10.1103/PhysRevResearch.2.012035} {\bibfield  {journal} {\bibinfo  {journal} {Phys. Rev. Res.}\ }\textbf {\bibinfo {volume} {2}},\ \bibinfo {pages} {012035} (\bibinfo {year} {2020})}\BibitemShut {NoStop}%
\bibitem [{\citenamefont {Hsieh}\ \emph {et~al.}(2020)\citenamefont {Hsieh}, \citenamefont {Lostaglio},\ and\ \citenamefont {Ac\'{\i}n}}]{Hsieh2020PRR}%
  \BibitemOpen
  \bibfield  {author} {\bibinfo {author} {\bibfnamefont {C.-Y.}\ \bibnamefont {Hsieh}}, \bibinfo {author} {\bibfnamefont {M.}~\bibnamefont {Lostaglio}},\ and\ \bibinfo {author} {\bibfnamefont {A.}~\bibnamefont {Ac\'{\i}n}},\ }\bibfield  {title} {\bibinfo {title} {Entanglement preserving local thermalization},\ }\href {https://doi.org/10.1103/PhysRevResearch.2.013379} {\bibfield  {journal} {\bibinfo  {journal} {Phys. Rev. Res.}\ }\textbf {\bibinfo {volume} {2}},\ \bibinfo {pages} {013379} (\bibinfo {year} {2020})}\BibitemShut {NoStop}%
\bibitem [{\citenamefont {Streltsov}\ \emph {et~al.}(2015)\citenamefont {Streltsov}, \citenamefont {Singh}, \citenamefont {Dhar}, \citenamefont {Bera},\ and\ \citenamefont {Adesso}}]{Streltsov2015PRL}%
  \BibitemOpen
  \bibfield  {author} {\bibinfo {author} {\bibfnamefont {A.}~\bibnamefont {Streltsov}}, \bibinfo {author} {\bibfnamefont {U.}~\bibnamefont {Singh}}, \bibinfo {author} {\bibfnamefont {H.~S.}\ \bibnamefont {Dhar}}, \bibinfo {author} {\bibfnamefont {M.~N.}\ \bibnamefont {Bera}},\ and\ \bibinfo {author} {\bibfnamefont {G.}~\bibnamefont {Adesso}},\ }\bibfield  {title} {\bibinfo {title} {Measuring quantum coherence with entanglement},\ }\href {https://doi.org/10.1103/PhysRevLett.115.020403} {\bibfield  {journal} {\bibinfo  {journal} {Phys. Rev. Lett.}\ }\textbf {\bibinfo {volume} {115}},\ \bibinfo {pages} {020403} (\bibinfo {year} {2015})}\BibitemShut {NoStop}%
\bibitem [{\citenamefont {Stratton}\ \emph {et~al.}(2024)\citenamefont {Stratton}, \citenamefont {Hsieh},\ and\ \citenamefont {Skrzypczyk}}]{Stratton2023}%
  \BibitemOpen
  \bibfield  {author} {\bibinfo {author} {\bibfnamefont {B.}~\bibnamefont {Stratton}}, \bibinfo {author} {\bibfnamefont {C.-Y.}\ \bibnamefont {Hsieh}},\ and\ \bibinfo {author} {\bibfnamefont {P.}~\bibnamefont {Skrzypczyk}},\ }\bibfield  {title} {\bibinfo {title} {Dynamical resource theory of informational nonequilibrium preservability},\ }\href {https://doi.org/10.1103/PhysRevLett.132.110202} {\bibfield  {journal} {\bibinfo  {journal} {Phys. Rev. Lett.}\ }\textbf {\bibinfo {volume} {132}},\ \bibinfo {pages} {110202} (\bibinfo {year} {2024})}\BibitemShut {NoStop}%
\bibitem [{\citenamefont {Lostaglio}(2019)}]{Lostaglio2019}%
  \BibitemOpen
  \bibfield  {author} {\bibinfo {author} {\bibfnamefont {M.}~\bibnamefont {Lostaglio}},\ }\bibfield  {title} {\bibinfo {title} {An introductory review of the resource theory approach to thermodynamics},\ }\href {https://doi.org/10.1088/1361-6633/ab46e5} {\bibfield  {journal} {\bibinfo  {journal} {Rep. Prog. Phys.}\ }\textbf {\bibinfo {volume} {82}},\ \bibinfo {pages} {114001} (\bibinfo {year} {2019})}\BibitemShut {NoStop}%
\bibitem [{\citenamefont {Faist}\ \emph {et~al.}(2015)\citenamefont {Faist}, \citenamefont {Oppenheim},\ and\ \citenamefont {Renner}}]{Faist2015NJP}%
  \BibitemOpen
  \bibfield  {author} {\bibinfo {author} {\bibfnamefont {P.}~\bibnamefont {Faist}}, \bibinfo {author} {\bibfnamefont {J.}~\bibnamefont {Oppenheim}},\ and\ \bibinfo {author} {\bibfnamefont {R.}~\bibnamefont {Renner}},\ }\bibfield  {title} {\bibinfo {title} {Gibbs-preserving maps outperform thermal operations in the quantum regime},\ }\href {https://doi.org/10.1088/1367-2630/17/4/043003} {\bibfield  {journal} {\bibinfo  {journal} {New J. Phys.}\ }\textbf {\bibinfo {volume} {17}},\ \bibinfo {pages} {043003} (\bibinfo {year} {2015})}\BibitemShut {NoStop}%
\bibitem [{\citenamefont {Horodecki}\ \emph {et~al.}(2003{\natexlab{a}})\citenamefont {Horodecki}, \citenamefont {Horodecki}, \citenamefont {Horodecki}, \citenamefont {Horodecki}, \citenamefont {Oppenheim}, \citenamefont {Sen(De)},\ and\ \citenamefont {Sen}}]{HorodeckiPRL2003}%
  \BibitemOpen
  \bibfield  {author} {\bibinfo {author} {\bibfnamefont {M.}~\bibnamefont {Horodecki}}, \bibinfo {author} {\bibfnamefont {K.}~\bibnamefont {Horodecki}}, \bibinfo {author} {\bibfnamefont {P.}~\bibnamefont {Horodecki}}, \bibinfo {author} {\bibfnamefont {R.}~\bibnamefont {Horodecki}}, \bibinfo {author} {\bibfnamefont {J.}~\bibnamefont {Oppenheim}}, \bibinfo {author} {\bibfnamefont {A.}~\bibnamefont {Sen(De)}},\ and\ \bibinfo {author} {\bibfnamefont {U.}~\bibnamefont {Sen}},\ }\bibfield  {title} {\bibinfo {title} {Local information as a resource in distributed quantum systems},\ }\href {https://doi.org/10.1103/PhysRevLett.90.100402} {\bibfield  {journal} {\bibinfo  {journal} {Phys. Rev. Lett.}\ }\textbf {\bibinfo {volume} {90}},\ \bibinfo {pages} {100402} (\bibinfo {year} {2003}{\natexlab{a}})}\BibitemShut {NoStop}%
\bibitem [{\citenamefont {Horodecki}\ \emph {et~al.}(2003{\natexlab{b}})\citenamefont {Horodecki}, \citenamefont {Horodecki},\ and\ \citenamefont {Oppenheim}}]{HorodeckiPRA2003}%
  \BibitemOpen
  \bibfield  {author} {\bibinfo {author} {\bibfnamefont {M.}~\bibnamefont {Horodecki}}, \bibinfo {author} {\bibfnamefont {P.}~\bibnamefont {Horodecki}},\ and\ \bibinfo {author} {\bibfnamefont {J.}~\bibnamefont {Oppenheim}},\ }\bibfield  {title} {\bibinfo {title} {Reversible transformations from pure to mixed states and the unique measure of information},\ }\href {https://doi.org/10.1103/PhysRevA.67.062104} {\bibfield  {journal} {\bibinfo  {journal} {Phys. Rev. A}\ }\textbf {\bibinfo {volume} {67}},\ \bibinfo {pages} {062104} (\bibinfo {year} {2003}{\natexlab{b}})}\BibitemShut {NoStop}%
\bibitem [{\citenamefont {Chitambar}\ and\ \citenamefont {Gour}(2019)}]{ChitambarRMP2019}%
  \BibitemOpen
  \bibfield  {author} {\bibinfo {author} {\bibfnamefont {E.}~\bibnamefont {Chitambar}}\ and\ \bibinfo {author} {\bibfnamefont {G.}~\bibnamefont {Gour}},\ }\bibfield  {title} {\bibinfo {title} {Quantum resource theories},\ }\href {https://doi.org/10.1103/RevModPhys.91.025001} {\bibfield  {journal} {\bibinfo  {journal} {Rev. Mod. Phys.}\ }\textbf {\bibinfo {volume} {91}},\ \bibinfo {pages} {025001} (\bibinfo {year} {2019})}\BibitemShut {NoStop}%
\bibitem [{\citenamefont {Piani}\ and\ \citenamefont {Watrous}(2009)}]{Piani2009PRL}%
  \BibitemOpen
  \bibfield  {author} {\bibinfo {author} {\bibfnamefont {M.}~\bibnamefont {Piani}}\ and\ \bibinfo {author} {\bibfnamefont {J.}~\bibnamefont {Watrous}},\ }\bibfield  {title} {\bibinfo {title} {All entangled states are useful for channel discrimination},\ }\href {https://doi.org/10.1103/PhysRevLett.102.250501} {\bibfield  {journal} {\bibinfo  {journal} {Phys. Rev. Lett.}\ }\textbf {\bibinfo {volume} {102}},\ \bibinfo {pages} {250501} (\bibinfo {year} {2009})}\BibitemShut {NoStop}%
\bibitem [{\citenamefont {Takagi}\ \emph {et~al.}(2019)\citenamefont {Takagi}, \citenamefont {Regula}, \citenamefont {Bu}, \citenamefont {Liu},\ and\ \citenamefont {Adesso}}]{Takagi2019PRL}%
  \BibitemOpen
  \bibfield  {author} {\bibinfo {author} {\bibfnamefont {R.}~\bibnamefont {Takagi}}, \bibinfo {author} {\bibfnamefont {B.}~\bibnamefont {Regula}}, \bibinfo {author} {\bibfnamefont {K.}~\bibnamefont {Bu}}, \bibinfo {author} {\bibfnamefont {Z.-W.}\ \bibnamefont {Liu}},\ and\ \bibinfo {author} {\bibfnamefont {G.}~\bibnamefont {Adesso}},\ }\bibfield  {title} {\bibinfo {title} {Operational advantage of quantum resources in subchannel discrimination},\ }\href {https://doi.org/10.1103/PhysRevLett.122.140402} {\bibfield  {journal} {\bibinfo  {journal} {Phys. Rev. Lett.}\ }\textbf {\bibinfo {volume} {122}},\ \bibinfo {pages} {140402} (\bibinfo {year} {2019})}\BibitemShut {NoStop}%
\bibitem [{\citenamefont {Skrzypczyk}\ and\ \citenamefont {Linden}(2019)}]{Skrzypczyk2019PRL}%
  \BibitemOpen
  \bibfield  {author} {\bibinfo {author} {\bibfnamefont {P.}~\bibnamefont {Skrzypczyk}}\ and\ \bibinfo {author} {\bibfnamefont {N.}~\bibnamefont {Linden}},\ }\bibfield  {title} {\bibinfo {title} {Robustness of measurement, discrimination games, and accessible information},\ }\href {https://doi.org/10.1103/PhysRevLett.122.140403} {\bibfield  {journal} {\bibinfo  {journal} {Phys. Rev. Lett.}\ }\textbf {\bibinfo {volume} {122}},\ \bibinfo {pages} {140403} (\bibinfo {year} {2019})}\BibitemShut {NoStop}%
\bibitem [{\citenamefont {Bae}\ \emph {et~al.}(2019)\citenamefont {Bae}, \citenamefont {Chru\ifmmode \acute{s}\else \'{s}\fi{}ci\ifmmode~\acute{n}\else \'{n}\fi{}ski},\ and\ \citenamefont {Piani}}]{Bae2019PRL}%
  \BibitemOpen
  \bibfield  {author} {\bibinfo {author} {\bibfnamefont {J.}~\bibnamefont {Bae}}, \bibinfo {author} {\bibfnamefont {D.}~\bibnamefont {Chru\ifmmode \acute{s}\else \'{s}\fi{}ci\ifmmode~\acute{n}\else \'{n}\fi{}ski}},\ and\ \bibinfo {author} {\bibfnamefont {M.}~\bibnamefont {Piani}},\ }\bibfield  {title} {\bibinfo {title} {More entanglement implies higher performance in channel discrimination tasks},\ }\href {https://doi.org/10.1103/PhysRevLett.122.140404} {\bibfield  {journal} {\bibinfo  {journal} {Phys. Rev. Lett.}\ }\textbf {\bibinfo {volume} {122}},\ \bibinfo {pages} {140404} (\bibinfo {year} {2019})}\BibitemShut {NoStop}%
\bibitem [{\citenamefont {Takagi}\ and\ \citenamefont {Regula}(2019)}]{Takagi2019}%
  \BibitemOpen
  \bibfield  {author} {\bibinfo {author} {\bibfnamefont {R.}~\bibnamefont {Takagi}}\ and\ \bibinfo {author} {\bibfnamefont {B.}~\bibnamefont {Regula}},\ }\bibfield  {title} {\bibinfo {title} {General resource theories in quantum mechanics and beyond: Operational characterization via discrimination tasks},\ }\href {https://doi.org/10.1103/PhysRevX.9.031053} {\bibfield  {journal} {\bibinfo  {journal} {Phys. Rev. X}\ }\textbf {\bibinfo {volume} {9}},\ \bibinfo {pages} {031053} (\bibinfo {year} {2019})}\BibitemShut {NoStop}%
\bibitem [{\citenamefont {Skrzypczyk}\ \emph {et~al.}(2019)\citenamefont {Skrzypczyk}, \citenamefont {\ifmmode \check{S}\else \v{S}\fi{}upi\ifmmode~\acute{c}\else \'{c}\fi{}},\ and\ \citenamefont {Cavalcanti}}]{Skrzypczyk2019}%
  \BibitemOpen
  \bibfield  {author} {\bibinfo {author} {\bibfnamefont {P.}~\bibnamefont {Skrzypczyk}}, \bibinfo {author} {\bibfnamefont {I.}~\bibnamefont {\ifmmode \check{S}\else \v{S}\fi{}upi\ifmmode~\acute{c}\else \'{c}\fi{}}},\ and\ \bibinfo {author} {\bibfnamefont {D.}~\bibnamefont {Cavalcanti}},\ }\bibfield  {title} {\bibinfo {title} {All sets of incompatible measurements give an advantage in quantum state discrimination},\ }\href {https://doi.org/10.1103/PhysRevLett.122.130403} {\bibfield  {journal} {\bibinfo  {journal} {Phys. Rev. Lett.}\ }\textbf {\bibinfo {volume} {122}},\ \bibinfo {pages} {130403} (\bibinfo {year} {2019})}\BibitemShut {NoStop}%
\bibitem [{\citenamefont {Hsieh}\ \emph {et~al.}(2024)\citenamefont {Hsieh}, \citenamefont {Tabia}, \citenamefont {Yin},\ and\ \citenamefont {Liang}}]{Hsieh2023-2}%
  \BibitemOpen
  \bibfield  {author} {\bibinfo {author} {\bibfnamefont {C.-Y.}\ \bibnamefont {Hsieh}}, \bibinfo {author} {\bibfnamefont {G.~N.~M.}\ \bibnamefont {Tabia}}, \bibinfo {author} {\bibfnamefont {Y.-C.}\ \bibnamefont {Yin}},\ and\ \bibinfo {author} {\bibfnamefont {Y.-C.}\ \bibnamefont {Liang}},\ }\bibfield  {title} {\bibinfo {title} {Resource marginal problems},\ }\href {https://doi.org/10.22331/q-2024-05-22-1353} {\bibfield  {journal} {\bibinfo  {journal} {{Quantum}}\ }\textbf {\bibinfo {volume} {8}},\ \bibinfo {pages} {1353} (\bibinfo {year} {2024})}\BibitemShut {NoStop}%
\bibitem [{\citenamefont {Hsieh}\ \emph {et~al.}(2022)\citenamefont {Hsieh}, \citenamefont {Lostaglio},\ and\ \citenamefont {Ac\'{\i}n}}]{Hsieh2022PRR}%
  \BibitemOpen
  \bibfield  {author} {\bibinfo {author} {\bibfnamefont {C.-Y.}\ \bibnamefont {Hsieh}}, \bibinfo {author} {\bibfnamefont {M.}~\bibnamefont {Lostaglio}},\ and\ \bibinfo {author} {\bibfnamefont {A.}~\bibnamefont {Ac\'{\i}n}},\ }\bibfield  {title} {\bibinfo {title} {Quantum channel marginal problem},\ }\href {https://doi.org/10.1103/PhysRevResearch.4.013249} {\bibfield  {journal} {\bibinfo  {journal} {Phys. Rev. Res.}\ }\textbf {\bibinfo {volume} {4}},\ \bibinfo {pages} {013249} (\bibinfo {year} {2022})}\BibitemShut {NoStop}%
\bibitem [{\citenamefont {Ducuara}\ and\ \citenamefont {Skrzypczyk}(2020)}]{Ducuara2020PRL}%
  \BibitemOpen
  \bibfield  {author} {\bibinfo {author} {\bibfnamefont {A.~F.}\ \bibnamefont {Ducuara}}\ and\ \bibinfo {author} {\bibfnamefont {P.}~\bibnamefont {Skrzypczyk}},\ }\bibfield  {title} {\bibinfo {title} {Operational interpretation of weight-based resource quantifiers in convex quantum resource theories},\ }\href {https://doi.org/10.1103/PhysRevLett.125.110401} {\bibfield  {journal} {\bibinfo  {journal} {Phys. Rev. Lett.}\ }\textbf {\bibinfo {volume} {125}},\ \bibinfo {pages} {110401} (\bibinfo {year} {2020})}\BibitemShut {NoStop}%
\bibitem [{\citenamefont {Uola}\ \emph {et~al.}(2020{\natexlab{b}})\citenamefont {Uola}, \citenamefont {Bullock}, \citenamefont {Kraft}, \citenamefont {Pellonp\"a\"a},\ and\ \citenamefont {Brunner}}]{Uola2020PRL}%
  \BibitemOpen
  \bibfield  {author} {\bibinfo {author} {\bibfnamefont {R.}~\bibnamefont {Uola}}, \bibinfo {author} {\bibfnamefont {T.}~\bibnamefont {Bullock}}, \bibinfo {author} {\bibfnamefont {T.}~\bibnamefont {Kraft}}, \bibinfo {author} {\bibfnamefont {J.-P.}\ \bibnamefont {Pellonp\"a\"a}},\ and\ \bibinfo {author} {\bibfnamefont {N.}~\bibnamefont {Brunner}},\ }\bibfield  {title} {\bibinfo {title} {All quantum resources provide an advantage in exclusion tasks},\ }\href {https://doi.org/10.1103/PhysRevLett.125.110402} {\bibfield  {journal} {\bibinfo  {journal} {Phys. Rev. Lett.}\ }\textbf {\bibinfo {volume} {125}},\ \bibinfo {pages} {110402} (\bibinfo {year} {2020}{\natexlab{b}})}\BibitemShut {NoStop}%
\bibitem [{\citenamefont {Tan}\ \emph {et~al.}(2021)\citenamefont {Tan}, \citenamefont {Narasimhachar},\ and\ \citenamefont {Regula}}]{Tan2021PRL}%
  \BibitemOpen
  \bibfield  {author} {\bibinfo {author} {\bibfnamefont {K.~C.}\ \bibnamefont {Tan}}, \bibinfo {author} {\bibfnamefont {V.}~\bibnamefont {Narasimhachar}},\ and\ \bibinfo {author} {\bibfnamefont {B.}~\bibnamefont {Regula}},\ }\bibfield  {title} {\bibinfo {title} {Fisher information universally identifies quantum resources},\ }\href {https://doi.org/10.1103/PhysRevLett.127.200402} {\bibfield  {journal} {\bibinfo  {journal} {Phys. Rev. Lett.}\ }\textbf {\bibinfo {volume} {127}},\ \bibinfo {pages} {200402} (\bibinfo {year} {2021})}\BibitemShut {NoStop}%
\bibitem [{\citenamefont {Uola}\ \emph {et~al.}(2020{\natexlab{c}})\citenamefont {Uola}, \citenamefont {Kraft},\ and\ \citenamefont {Abbott}}]{Uola2020PRA}%
  \BibitemOpen
  \bibfield  {author} {\bibinfo {author} {\bibfnamefont {R.}~\bibnamefont {Uola}}, \bibinfo {author} {\bibfnamefont {T.}~\bibnamefont {Kraft}},\ and\ \bibinfo {author} {\bibfnamefont {A.~A.}\ \bibnamefont {Abbott}},\ }\bibfield  {title} {\bibinfo {title} {Quantification of quantum dynamics with input-output games},\ }\href {https://doi.org/10.1103/PhysRevA.101.052306} {\bibfield  {journal} {\bibinfo  {journal} {Phys. Rev. A}\ }\textbf {\bibinfo {volume} {101}},\ \bibinfo {pages} {052306} (\bibinfo {year} {2020}{\natexlab{c}})}\BibitemShut {NoStop}%
\bibitem [{Note1()}]{Note1}%
  \BibitemOpen
  \bibinfo {note} {{See also Ref.~\cite {Rosset2020PRL} for spacelike separated resources}}\BibitemShut {NoStop}%
\bibitem [{\citenamefont {Kosloff}\ and\ \citenamefont {Feldmann}(2002)}]{KosloffPRE2002}%
  \BibitemOpen
  \bibfield  {author} {\bibinfo {author} {\bibfnamefont {R.}~\bibnamefont {Kosloff}}\ and\ \bibinfo {author} {\bibfnamefont {T.}~\bibnamefont {Feldmann}},\ }\bibfield  {title} {\bibinfo {title} {Discrete four-stroke quantum heat engine exploring the origin of friction},\ }\href {https://doi.org/10.1103/PhysRevE.65.055102} {\bibfield  {journal} {\bibinfo  {journal} {Phys. Rev. E}\ }\textbf {\bibinfo {volume} {65}},\ \bibinfo {pages} {055102} (\bibinfo {year} {2002})}\BibitemShut {NoStop}%
\bibitem [{\citenamefont {Feldmann}\ and\ \citenamefont {Kosloff}(2006)}]{FeldmannPRE2006}%
  \BibitemOpen
  \bibfield  {author} {\bibinfo {author} {\bibfnamefont {T.}~\bibnamefont {Feldmann}}\ and\ \bibinfo {author} {\bibfnamefont {R.}~\bibnamefont {Kosloff}},\ }\bibfield  {title} {\bibinfo {title} {Quantum lubrication: Suppression of friction in a first-principles four-stroke heat engine},\ }\href {https://doi.org/10.1103/PhysRevE.73.025107} {\bibfield  {journal} {\bibinfo  {journal} {Phys. Rev. E}\ }\textbf {\bibinfo {volume} {73}},\ \bibinfo {pages} {025107} (\bibinfo {year} {2006})}\BibitemShut {NoStop}%
\bibitem [{\citenamefont {Ji}\ \emph {et~al.}(2022)\citenamefont {Ji}, \citenamefont {Chai}, \citenamefont {Wang}, \citenamefont {Guo}, \citenamefont {Rong}, \citenamefont {Shi}, \citenamefont {Ren}, \citenamefont {Wang},\ and\ \citenamefont {Du}}]{JiPRL2022}%
  \BibitemOpen
  \bibfield  {author} {\bibinfo {author} {\bibfnamefont {W.}~\bibnamefont {Ji}}, \bibinfo {author} {\bibfnamefont {Z.}~\bibnamefont {Chai}}, \bibinfo {author} {\bibfnamefont {M.}~\bibnamefont {Wang}}, \bibinfo {author} {\bibfnamefont {Y.}~\bibnamefont {Guo}}, \bibinfo {author} {\bibfnamefont {X.}~\bibnamefont {Rong}}, \bibinfo {author} {\bibfnamefont {F.}~\bibnamefont {Shi}}, \bibinfo {author} {\bibfnamefont {C.}~\bibnamefont {Ren}}, \bibinfo {author} {\bibfnamefont {Y.}~\bibnamefont {Wang}},\ and\ \bibinfo {author} {\bibfnamefont {J.}~\bibnamefont {Du}},\ }\bibfield  {title} {\bibinfo {title} {Spin quantum heat engine quantified by quantum steering},\ }\href {https://doi.org/10.1103/PhysRevLett.128.090602} {\bibfield  {journal} {\bibinfo  {journal} {Phys. Rev. Lett.}\ }\textbf {\bibinfo {volume} {128}},\ \bibinfo {pages} {090602} (\bibinfo {year} {2022})}\BibitemShut {NoStop}%
\bibitem [{\citenamefont {Beyer}\ \emph {et~al.}(2019)\citenamefont {Beyer}, \citenamefont {Luoma},\ and\ \citenamefont {Strunz}}]{BeyerPRL2019}%
  \BibitemOpen
  \bibfield  {author} {\bibinfo {author} {\bibfnamefont {K.}~\bibnamefont {Beyer}}, \bibinfo {author} {\bibfnamefont {K.}~\bibnamefont {Luoma}},\ and\ \bibinfo {author} {\bibfnamefont {W.~T.}\ \bibnamefont {Strunz}},\ }\bibfield  {title} {\bibinfo {title} {Steering heat engines: A truly quantum maxwell demon},\ }\href {https://doi.org/10.1103/PhysRevLett.123.250606} {\bibfield  {journal} {\bibinfo  {journal} {Phys. Rev. Lett.}\ }\textbf {\bibinfo {volume} {123}},\ \bibinfo {pages} {250606} (\bibinfo {year} {2019})}\BibitemShut {NoStop}%
\bibitem [{\citenamefont {Chan}\ \emph {et~al.}(2022)\citenamefont {Chan}, \citenamefont {Huang}, \citenamefont {Lin}, \citenamefont {Ku}, \citenamefont {Chen}, \citenamefont {Chen},\ and\ \citenamefont {Chen}}]{ChanPRA2022}%
  \BibitemOpen
  \bibfield  {author} {\bibinfo {author} {\bibfnamefont {F.-J.}\ \bibnamefont {Chan}}, \bibinfo {author} {\bibfnamefont {Y.-T.}\ \bibnamefont {Huang}}, \bibinfo {author} {\bibfnamefont {J.-D.}\ \bibnamefont {Lin}}, \bibinfo {author} {\bibfnamefont {H.-Y.}\ \bibnamefont {Ku}}, \bibinfo {author} {\bibfnamefont {J.-S.}\ \bibnamefont {Chen}}, \bibinfo {author} {\bibfnamefont {H.-B.}\ \bibnamefont {Chen}},\ and\ \bibinfo {author} {\bibfnamefont {Y.-N.}\ \bibnamefont {Chen}},\ }\bibfield  {title} {\bibinfo {title} {Maxwell's two-demon engine under pure dephasing noise},\ }\href {https://doi.org/10.1103/PhysRevA.106.052201} {\bibfield  {journal} {\bibinfo  {journal} {Phys. Rev. A}\ }\textbf {\bibinfo {volume} {106}},\ \bibinfo {pages} {052201} (\bibinfo {year} {2022})}\BibitemShut {NoStop}%
\bibitem [{\citenamefont {Biswas}\ \emph {et~al.}(2025)\citenamefont {Biswas}, \citenamefont {Datta},\ and\ \citenamefont {Garc\'{\i}a-Pintos}}]{Biswas2025PRL}%
  \BibitemOpen
  \bibfield  {author} {\bibinfo {author} {\bibfnamefont {T.}~\bibnamefont {Biswas}}, \bibinfo {author} {\bibfnamefont {C.}~\bibnamefont {Datta}},\ and\ \bibinfo {author} {\bibfnamefont {L.~P.}\ \bibnamefont {Garc\'{\i}a-Pintos}},\ }\bibfield  {title} {\bibinfo {title} {Quantum thermodynamic advantage in work extraction from steerable quantum correlations},\ }\href {https://doi.org/10.1103/9qcc-7lq5} {\bibfield  {journal} {\bibinfo  {journal} {Phys. Rev. Lett.}\ }\textbf {\bibinfo {volume} {135}},\ \bibinfo {pages} {110402} (\bibinfo {year} {2025})}\BibitemShut {NoStop}%
\bibitem [{\citenamefont {Jennings}\ and\ \citenamefont {Rudolph}(2010)}]{Jennings2010PRE}%
  \BibitemOpen
  \bibfield  {author} {\bibinfo {author} {\bibfnamefont {D.}~\bibnamefont {Jennings}}\ and\ \bibinfo {author} {\bibfnamefont {T.}~\bibnamefont {Rudolph}},\ }\bibfield  {title} {\bibinfo {title} {Entanglement and the thermodynamic arrow of time},\ }\href {https://doi.org/10.1103/PhysRevE.81.061130} {\bibfield  {journal} {\bibinfo  {journal} {Phys. Rev. E}\ }\textbf {\bibinfo {volume} {81}},\ \bibinfo {pages} {061130} (\bibinfo {year} {2010})}\BibitemShut {NoStop}%
\bibitem [{\citenamefont {Lostaglio}\ \emph {et~al.}(2017)\citenamefont {Lostaglio}, \citenamefont {Jennings},\ and\ \citenamefont {Rudolph}}]{LostaglioNJP2017}%
  \BibitemOpen
  \bibfield  {author} {\bibinfo {author} {\bibfnamefont {M.}~\bibnamefont {Lostaglio}}, \bibinfo {author} {\bibfnamefont {D.}~\bibnamefont {Jennings}},\ and\ \bibinfo {author} {\bibfnamefont {T.}~\bibnamefont {Rudolph}},\ }\bibfield  {title} {\bibinfo {title} {Thermodynamic resource theories, non-commutativity and maximum entropy principles},\ }\href {https://doi.org/10.1088/1367-2630/aa617f} {\bibfield  {journal} {\bibinfo  {journal} {New J. Phys.}\ }\textbf {\bibinfo {volume} {19}},\ \bibinfo {pages} {043008} (\bibinfo {year} {2017})}\BibitemShut {NoStop}%
\bibitem [{\citenamefont {Majidy}\ \emph {et~al.}(2023)\citenamefont {Majidy}, \citenamefont {Braasch}, \citenamefont {Lasek}, \citenamefont {Upadhyaya}, \citenamefont {Kalev},\ and\ \citenamefont {Yunger~Halpern}}]{Majidy2023}%
  \BibitemOpen
  \bibfield  {author} {\bibinfo {author} {\bibfnamefont {S.}~\bibnamefont {Majidy}}, \bibinfo {author} {\bibfnamefont {W.~F.}\ \bibnamefont {Braasch}}, \bibinfo {author} {\bibfnamefont {A.}~\bibnamefont {Lasek}}, \bibinfo {author} {\bibfnamefont {T.}~\bibnamefont {Upadhyaya}}, \bibinfo {author} {\bibfnamefont {A.}~\bibnamefont {Kalev}},\ and\ \bibinfo {author} {\bibfnamefont {N.}~\bibnamefont {Yunger~Halpern}},\ }\bibfield  {title} {\bibinfo {title} {Noncommuting conserved charges in quantum thermodynamics and beyond},\ }\href {https://doi.org/10.1038/s42254-023-00641-9} {\bibfield  {journal} {\bibinfo  {journal} {Nat. Rev. Phys.}\ }\textbf {\bibinfo {volume} {5}},\ \bibinfo {pages} {689} (\bibinfo {year} {2023})}\BibitemShut {NoStop}%
\bibitem [{\citenamefont {Yunger~Halpern}\ \emph {et~al.}(2016)\citenamefont {Yunger~Halpern}, \citenamefont {Faist}, \citenamefont {Oppenheim},\ and\ \citenamefont {Winter}}]{YungerHalpern2016NC}%
  \BibitemOpen
  \bibfield  {author} {\bibinfo {author} {\bibfnamefont {N.}~\bibnamefont {Yunger~Halpern}}, \bibinfo {author} {\bibfnamefont {P.}~\bibnamefont {Faist}}, \bibinfo {author} {\bibfnamefont {J.}~\bibnamefont {Oppenheim}},\ and\ \bibinfo {author} {\bibfnamefont {A.}~\bibnamefont {Winter}},\ }\bibfield  {title} {\bibinfo {title} {Microcanonical and resource-theoretic derivations of the thermal state of a quantum system with noncommuting charges},\ }\href {https://doi.org/10.1038/ncomms12051} {\bibfield  {journal} {\bibinfo  {journal} {Nat. Commun.}\ }\textbf {\bibinfo {volume} {7}},\ \bibinfo {pages} {12051} (\bibinfo {year} {2016})}\BibitemShut {NoStop}%
\bibitem [{\citenamefont {Guryanova}\ \emph {et~al.}(2016)\citenamefont {Guryanova}, \citenamefont {Popescu}, \citenamefont {Short}, \citenamefont {Silva},\ and\ \citenamefont {Skrzypczyk}}]{Guryanova2016NC}%
  \BibitemOpen
  \bibfield  {author} {\bibinfo {author} {\bibfnamefont {Y.}~\bibnamefont {Guryanova}}, \bibinfo {author} {\bibfnamefont {S.}~\bibnamefont {Popescu}}, \bibinfo {author} {\bibfnamefont {A.~J.}\ \bibnamefont {Short}}, \bibinfo {author} {\bibfnamefont {R.}~\bibnamefont {Silva}},\ and\ \bibinfo {author} {\bibfnamefont {P.}~\bibnamefont {Skrzypczyk}},\ }\bibfield  {title} {\bibinfo {title} {Thermodynamics of quantum systems with multiple conserved quantities},\ }\href {https://doi.org/10.1038/ncomms12049} {\bibfield  {journal} {\bibinfo  {journal} {Nat. Commun.}\ }\textbf {\bibinfo {volume} {7}},\ \bibinfo {pages} {12049} (\bibinfo {year} {2016})}\BibitemShut {NoStop}%
\bibitem [{\citenamefont {Lostaglio}(2020)}]{Lostaglio2020PRL}%
  \BibitemOpen
  \bibfield  {author} {\bibinfo {author} {\bibfnamefont {M.}~\bibnamefont {Lostaglio}},\ }\bibfield  {title} {\bibinfo {title} {Certifying quantum signatures in thermodynamics and metrology via contextuality of quantum linear response},\ }\href {https://doi.org/10.1103/PhysRevLett.125.230603} {\bibfield  {journal} {\bibinfo  {journal} {Phys. Rev. Lett.}\ }\textbf {\bibinfo {volume} {125}},\ \bibinfo {pages} {230603} (\bibinfo {year} {2020})}\BibitemShut {NoStop}%
\bibitem [{\citenamefont {Levy}\ and\ \citenamefont {Lostaglio}(2020)}]{Levy2020PRXQ}%
  \BibitemOpen
  \bibfield  {author} {\bibinfo {author} {\bibfnamefont {A.}~\bibnamefont {Levy}}\ and\ \bibinfo {author} {\bibfnamefont {M.}~\bibnamefont {Lostaglio}},\ }\bibfield  {title} {\bibinfo {title} {Quasiprobability distribution for heat fluctuations in the quantum regime},\ }\href {https://doi.org/10.1103/PRXQuantum.1.010309} {\bibfield  {journal} {\bibinfo  {journal} {PRX Quantum}\ }\textbf {\bibinfo {volume} {1}},\ \bibinfo {pages} {010309} (\bibinfo {year} {2020})}\BibitemShut {NoStop}%
\bibitem [{\citenamefont {Puliyil}\ \emph {et~al.}(2022)\citenamefont {Puliyil}, \citenamefont {Banik},\ and\ \citenamefont {Alimuddin}}]{Puliyil2022PRL}%
  \BibitemOpen
  \bibfield  {author} {\bibinfo {author} {\bibfnamefont {S.}~\bibnamefont {Puliyil}}, \bibinfo {author} {\bibfnamefont {M.}~\bibnamefont {Banik}},\ and\ \bibinfo {author} {\bibfnamefont {M.}~\bibnamefont {Alimuddin}},\ }\bibfield  {title} {\bibinfo {title} {Thermodynamic signatures of genuinely multipartite entanglement},\ }\href {https://doi.org/10.1103/PhysRevLett.129.070601} {\bibfield  {journal} {\bibinfo  {journal} {Phys. Rev. Lett.}\ }\textbf {\bibinfo {volume} {129}},\ \bibinfo {pages} {070601} (\bibinfo {year} {2022})}\BibitemShut {NoStop}%
\bibitem [{\citenamefont {Upadhyaya}\ \emph {et~al.}(2024)\citenamefont {Upadhyaya}, \citenamefont {Braasch}, \citenamefont {Landi},\ and\ \citenamefont {Halpern}}]{Upadhyaya2023}%
  \BibitemOpen
  \bibfield  {author} {\bibinfo {author} {\bibfnamefont {T.}~\bibnamefont {Upadhyaya}}, \bibinfo {author} {\bibfnamefont {W.~F.}\ \bibnamefont {Braasch}}, \bibinfo {author} {\bibfnamefont {G.~T.}\ \bibnamefont {Landi}},\ and\ \bibinfo {author} {\bibfnamefont {N.~Y.}\ \bibnamefont {Halpern}},\ }\bibfield  {title} {\bibinfo {title} {Non-abelian transport distinguishes three usually equivalent notions of entropy production},\ }\href {https://doi.org/10.1103/PRXQuantum.5.030355} {\bibfield  {journal} {\bibinfo  {journal} {PRX Quantum}\ }\textbf {\bibinfo {volume} {5}},\ \bibinfo {pages} {030355} (\bibinfo {year} {2024})}\BibitemShut {NoStop}%
\bibitem [{\citenamefont {Centrone}\ and\ \citenamefont {Gessner}(2024)}]{Centrone2024}%
  \BibitemOpen
  \bibfield  {author} {\bibinfo {author} {\bibfnamefont {F.}~\bibnamefont {Centrone}}\ and\ \bibinfo {author} {\bibfnamefont {M.}~\bibnamefont {Gessner}},\ }\bibfield  {title} {\bibinfo {title} {Breaking local quantum speed limits with steering},\ }\href {https://doi.org/10.1103/PhysRevResearch.6.L042067} {\bibfield  {journal} {\bibinfo  {journal} {Phys. Rev. Res.}\ }\textbf {\bibinfo {volume} {6}},\ \bibinfo {pages} {L042067} (\bibinfo {year} {2024})}\BibitemShut {NoStop}%
\bibitem [{\citenamefont {Lipka-Bartosik}\ \emph {et~al.}(2024)\citenamefont {Lipka-Bartosik}, \citenamefont {Diotallevi},\ and\ \citenamefont {Bakhshinezhad}}]{Lipkabartosik2023}%
  \BibitemOpen
  \bibfield  {author} {\bibinfo {author} {\bibfnamefont {P.}~\bibnamefont {Lipka-Bartosik}}, \bibinfo {author} {\bibfnamefont {G.~F.}\ \bibnamefont {Diotallevi}},\ and\ \bibinfo {author} {\bibfnamefont {P.}~\bibnamefont {Bakhshinezhad}},\ }\bibfield  {title} {\bibinfo {title} {Fundamental limits on anomalous energy flows in correlated quantum systems},\ }\href {https://doi.org/10.1103/PhysRevLett.132.140402} {\bibfield  {journal} {\bibinfo  {journal} {Phys. Rev. Lett.}\ }\textbf {\bibinfo {volume} {132}},\ \bibinfo {pages} {140402} (\bibinfo {year} {2024})}\BibitemShut {NoStop}%
\bibitem [{\citenamefont {Perarnau-Llobet}\ \emph {et~al.}(2015)\citenamefont {Perarnau-Llobet}, \citenamefont {Hovhannisyan}, \citenamefont {Huber}, \citenamefont {Skrzypczyk}, \citenamefont {Brunner},\ and\ \citenamefont {Ac\'{\i}n}}]{Perarnau-LlobetPRX2015}%
  \BibitemOpen
  \bibfield  {author} {\bibinfo {author} {\bibfnamefont {M.}~\bibnamefont {Perarnau-Llobet}}, \bibinfo {author} {\bibfnamefont {K.~V.}\ \bibnamefont {Hovhannisyan}}, \bibinfo {author} {\bibfnamefont {M.}~\bibnamefont {Huber}}, \bibinfo {author} {\bibfnamefont {P.}~\bibnamefont {Skrzypczyk}}, \bibinfo {author} {\bibfnamefont {N.}~\bibnamefont {Brunner}},\ and\ \bibinfo {author} {\bibfnamefont {A.}~\bibnamefont {Ac\'{\i}n}},\ }\bibfield  {title} {\bibinfo {title} {Extractable work from correlations},\ }\href {https://doi.org/10.1103/PhysRevX.5.041011} {\bibfield  {journal} {\bibinfo  {journal} {Phys. Rev. X}\ }\textbf {\bibinfo {volume} {5}},\ \bibinfo {pages} {041011} (\bibinfo {year} {2015})}\BibitemShut {NoStop}%
\bibitem [{\citenamefont {Rio}\ \emph {et~al.}(2011)\citenamefont {Rio}, \citenamefont {{\AA}berg}, \citenamefont {Renner}, \citenamefont {Dahlsten},\ and\ \citenamefont {Vedral}}]{Rio2011}%
  \BibitemOpen
  \bibfield  {author} {\bibinfo {author} {\bibfnamefont {L.~d.}\ \bibnamefont {Rio}}, \bibinfo {author} {\bibfnamefont {J.}~\bibnamefont {{\AA}berg}}, \bibinfo {author} {\bibfnamefont {R.}~\bibnamefont {Renner}}, \bibinfo {author} {\bibfnamefont {O.}~\bibnamefont {Dahlsten}},\ and\ \bibinfo {author} {\bibfnamefont {V.}~\bibnamefont {Vedral}},\ }\bibfield  {title} {\bibinfo {title} {The thermodynamic meaning of negative entropy},\ }\href {https://doi.org/10.1038/nature10123} {\bibfield  {journal} {\bibinfo  {journal} {Nature}\ }\textbf {\bibinfo {volume} {474}},\ \bibinfo {pages} {61} (\bibinfo {year} {2011})}\BibitemShut {NoStop}%
\bibitem [{\citenamefont {Skrzypczyk}\ \emph {et~al.}(2014)\citenamefont {Skrzypczyk}, \citenamefont {Short},\ and\ \citenamefont {Popescu}}]{Skrzypczyk2014NC}%
  \BibitemOpen
  \bibfield  {author} {\bibinfo {author} {\bibfnamefont {P.}~\bibnamefont {Skrzypczyk}}, \bibinfo {author} {\bibfnamefont {A.~J.}\ \bibnamefont {Short}},\ and\ \bibinfo {author} {\bibfnamefont {S.}~\bibnamefont {Popescu}},\ }\bibfield  {title} {\bibinfo {title} {Work extraction and thermodynamics for individual quantum systems},\ }\href {https://doi.org/10.1038/ncomms5185} {\bibfield  {journal} {\bibinfo  {journal} {Nat. Commun.}\ }\textbf {\bibinfo {volume} {5}},\ \bibinfo {pages} {4185} (\bibinfo {year} {2014})}\BibitemShut {NoStop}%
\bibitem [{\citenamefont {Shiraishi}\ and\ \citenamefont {Takagi}(2024)}]{Shiraishi2023}%
  \BibitemOpen
  \bibfield  {author} {\bibinfo {author} {\bibfnamefont {N.}~\bibnamefont {Shiraishi}}\ and\ \bibinfo {author} {\bibfnamefont {R.}~\bibnamefont {Takagi}},\ }\bibfield  {title} {\bibinfo {title} {Arbitrary amplification of quantum coherence in asymptotic and catalytic transformation},\ }\href {https://doi.org/10.1103/PhysRevLett.132.180202} {\bibfield  {journal} {\bibinfo  {journal} {Phys. Rev. Lett.}\ }\textbf {\bibinfo {volume} {132}},\ \bibinfo {pages} {180202} (\bibinfo {year} {2024})}\BibitemShut {NoStop}%
\bibitem [{\citenamefont {Korzekwa}\ \emph {et~al.}(2016)\citenamefont {Korzekwa}, \citenamefont {Lostaglio}, \citenamefont {Oppenheim},\ and\ \citenamefont {Jennings}}]{Korzekwa2016NJP}%
  \BibitemOpen
  \bibfield  {author} {\bibinfo {author} {\bibfnamefont {K.}~\bibnamefont {Korzekwa}}, \bibinfo {author} {\bibfnamefont {M.}~\bibnamefont {Lostaglio}}, \bibinfo {author} {\bibfnamefont {J.}~\bibnamefont {Oppenheim}},\ and\ \bibinfo {author} {\bibfnamefont {D.}~\bibnamefont {Jennings}},\ }\bibfield  {title} {\bibinfo {title} {The extraction of work from quantum coherence},\ }\href {https://doi.org/10.1088/1367-2630/18/2/023045} {\bibfield  {journal} {\bibinfo  {journal} {New J. Phys.}\ }\textbf {\bibinfo {volume} {18}},\ \bibinfo {pages} {023045} (\bibinfo {year} {2016})}\BibitemShut {NoStop}%
\bibitem [{\citenamefont {Hsieh}\ and\ \citenamefont {Chen}(2024)}]{Hsieh2024}%
  \BibitemOpen
  \bibfield  {author} {\bibinfo {author} {\bibfnamefont {C.-Y.}\ \bibnamefont {Hsieh}}\ and\ \bibinfo {author} {\bibfnamefont {S.-L.}\ \bibnamefont {Chen}},\ }\bibfield  {title} {\bibinfo {title} {Thermodynamic approach to quantifying incompatible instruments},\ }\href {https://doi.org/10.1103/PhysRevLett.133.170401} {\bibfield  {journal} {\bibinfo  {journal} {Phys. Rev. Lett.}\ }\textbf {\bibinfo {volume} {133}},\ \bibinfo {pages} {170401} (\bibinfo {year} {2024})}\BibitemShut {NoStop}%
\bibitem [{\citenamefont {de~Oliveira~Junior}\ \emph {et~al.}(2025)\citenamefont {de~Oliveira~Junior}, \citenamefont {Brask},\ and\ \citenamefont {Lipka-Bartosik}}]{deOliveiraJunior2025PRL}%
  \BibitemOpen
  \bibfield  {author} {\bibinfo {author} {\bibfnamefont {A.}~\bibnamefont {de~Oliveira~Junior}}, \bibinfo {author} {\bibfnamefont {J.~B.}\ \bibnamefont {Brask}},\ and\ \bibinfo {author} {\bibfnamefont {P.}~\bibnamefont {Lipka-Bartosik}},\ }\bibfield  {title} {\bibinfo {title} {Heat as a witness of quantum properties},\ }\href {https://doi.org/10.1103/PhysRevLett.134.050401} {\bibfield  {journal} {\bibinfo  {journal} {Phys. Rev. Lett.}\ }\textbf {\bibinfo {volume} {134}},\ \bibinfo {pages} {050401} (\bibinfo {year} {2025})}\BibitemShut {NoStop}%
\bibitem [{\citenamefont {Macêdo}\ \emph {et~al.}(2026)\citenamefont {Macêdo}, \citenamefont {de~Oliveira~Junior}, \citenamefont {Comar}, \citenamefont {Keller}, \citenamefont {Brask}, \citenamefont {Céleri},\ and\ \citenamefont {Chaves}}]{Macedo2026}%
  \BibitemOpen
  \bibfield  {author} {\bibinfo {author} {\bibfnamefont {R.~A.}\ \bibnamefont {Macêdo}}, \bibinfo {author} {\bibfnamefont {A.}~\bibnamefont {de~Oliveira~Junior}}, \bibinfo {author} {\bibfnamefont {N.~E.}\ \bibnamefont {Comar}}, \bibinfo {author} {\bibfnamefont {L.~L.}\ \bibnamefont {Keller}}, \bibinfo {author} {\bibfnamefont {J.~B.}\ \bibnamefont {Brask}}, \bibinfo {author} {\bibfnamefont {L.~C.}\ \bibnamefont {Céleri}},\ and\ \bibinfo {author} {\bibfnamefont {R.}~\bibnamefont {Chaves}},\ }\href@noop {} {\bibinfo {title} {Every little thing heat does is magic}} (\bibinfo {year} {2026}),\ \Eprint {https://arxiv.org/abs/2604.08663} {arXiv:2604.08663 [quant-ph]} \BibitemShut {NoStop}%
\bibitem [{\citenamefont {Buscemi}(2012)}]{Buscemi2012PRL}%
  \BibitemOpen
  \bibfield  {author} {\bibinfo {author} {\bibfnamefont {F.}~\bibnamefont {Buscemi}},\ }\bibfield  {title} {\bibinfo {title} {All entangled quantum states are nonlocal},\ }\href {https://doi.org/10.1103/PhysRevLett.108.200401} {\bibfield  {journal} {\bibinfo  {journal} {Phys. Rev. Lett.}\ }\textbf {\bibinfo {volume} {108}},\ \bibinfo {pages} {200401} (\bibinfo {year} {2012})}\BibitemShut {NoStop}%
\bibitem [{\citenamefont {Gour}\ \emph {et~al.}(2015)\citenamefont {Gour}, \citenamefont {Müller}, \citenamefont {Narasimhachar}, \citenamefont {Spekkens},\ and\ \citenamefont {{Yunger Halpern}}}]{Purity-review}%
  \BibitemOpen
  \bibfield  {author} {\bibinfo {author} {\bibfnamefont {G.}~\bibnamefont {Gour}}, \bibinfo {author} {\bibfnamefont {M.~P.}\ \bibnamefont {Müller}}, \bibinfo {author} {\bibfnamefont {V.}~\bibnamefont {Narasimhachar}}, \bibinfo {author} {\bibfnamefont {R.~W.}\ \bibnamefont {Spekkens}},\ and\ \bibinfo {author} {\bibfnamefont {N.}~\bibnamefont {{Yunger Halpern}}},\ }\bibfield  {title} {\bibinfo {title} {The resource theory of informational nonequilibrium in thermodynamics},\ }\href {https://doi.org/https://doi.org/10.1016/j.physrep.2015.04.003} {\bibfield  {journal} {\bibinfo  {journal} {Phys. Rep.}\ }\textbf {\bibinfo {volume} {583}},\ \bibinfo {pages} {1} (\bibinfo {year} {2015})}\BibitemShut {NoStop}%
\bibitem [{\citenamefont {Hsieh}\ and\ \citenamefont {Gessner}(2026)}]{Companion-arXiv}%
  \BibitemOpen
  \bibfield  {author} {\bibinfo {author} {\bibfnamefont {C.-Y.}\ \bibnamefont {Hsieh}}\ and\ \bibinfo {author} {\bibfnamefont {M.}~\bibnamefont {Gessner}},\ }\bibfield  {title} {\bibinfo {title} {companion paper, {G}eneral quantum resources providing advantages in work-extraction tasks},\ }\href {https://doi.org/10.1103/p5mg-s3dw} {\bibfield  {journal} {\bibinfo  {journal} {Phys. Rev. A}\ }\textbf {\bibinfo {volume} {114}},\ \bibinfo {pages} {022456} (\bibinfo {year} {2026})}\BibitemShut {NoStop}%
\bibitem [{\citenamefont {Nielsen}\ and\ \citenamefont {Chuang}(2010)}]{QIC-book}%
  \BibitemOpen
  \bibfield  {author} {\bibinfo {author} {\bibfnamefont {M.~A.}\ \bibnamefont {Nielsen}}\ and\ \bibinfo {author} {\bibfnamefont {I.~L.}\ \bibnamefont {Chuang}},\ }\href@noop {} {\emph {\bibinfo {title} {Quantum Computation and Quantum Information: 10th Anniversary Edition}}}\ (\bibinfo  {publisher} {Cambridge University Press},\ \bibinfo {year} {2010})\BibitemShut {NoStop}%
\bibitem [{Note2()}]{Note2}%
  \BibitemOpen
  \bibinfo {note} {We always assume $\protect \mathcal {O}_R$ is compact in the topology induced by the diamond distance, i.e., the distance measure induced by diamond norm $\left \|\cdot \right \|_\diamond $~\cite {Watrous-book} (see also, e.g., Ref.~\cite {Regula2021Quantum}).}\BibitemShut {Stop}%
\bibitem [{\citenamefont {Brandão}\ \emph {et~al.}(2015)\citenamefont {Brandão}, \citenamefont {Horodecki}, \citenamefont {Ng}, \citenamefont {Oppenheim},\ and\ \citenamefont {Wehner}}]{Brandão2015}%
  \BibitemOpen
  \bibfield  {author} {\bibinfo {author} {\bibfnamefont {F.}~\bibnamefont {Brandão}}, \bibinfo {author} {\bibfnamefont {M.}~\bibnamefont {Horodecki}}, \bibinfo {author} {\bibfnamefont {N.}~\bibnamefont {Ng}}, \bibinfo {author} {\bibfnamefont {J.}~\bibnamefont {Oppenheim}},\ and\ \bibinfo {author} {\bibfnamefont {S.}~\bibnamefont {Wehner}},\ }\bibfield  {title} {\bibinfo {title} {The second laws of quantum thermodynamics},\ }\href {https://doi.org/10.1073/pnas.1411728112} {\bibfield  {journal} {\bibinfo  {journal} {PNAS}\ }\textbf {\bibinfo {volume} {112}},\ \bibinfo {pages} {3275} (\bibinfo {year} {2015})}\BibitemShut {NoStop}%
\bibitem [{\citenamefont {Ćwikliński}\ \emph {et~al.}(2015)\citenamefont {Ćwikliński}, \citenamefont {Studziński}, \citenamefont {Horodecki},\ and\ \citenamefont {Oppenheim}}]{Ćwikliński2015PRL}%
  \BibitemOpen
  \bibfield  {author} {\bibinfo {author} {\bibfnamefont {P.}~\bibnamefont {Ćwikliński}}, \bibinfo {author} {\bibfnamefont {M.}~\bibnamefont {Studziński}}, \bibinfo {author} {\bibfnamefont {M.}~\bibnamefont {Horodecki}},\ and\ \bibinfo {author} {\bibfnamefont {J.}~\bibnamefont {Oppenheim}},\ }\bibfield  {title} {\bibinfo {title} {Limitations on the evolution of quantum coherences: Towards fully quantum second laws of thermodynamics},\ }\href {https://doi.org/10.1103/PhysRevLett.115.210403} {\bibfield  {journal} {\bibinfo  {journal} {Phys. Rev. Lett.}\ }\textbf {\bibinfo {volume} {115}},\ \bibinfo {pages} {210403} (\bibinfo {year} {2015})}\BibitemShut {NoStop}%
\bibitem [{\citenamefont {Gour}\ \emph {et~al.}(2018)\citenamefont {Gour}, \citenamefont {Jennings}, \citenamefont {Buscemi}, \citenamefont {Duan},\ and\ \citenamefont {Marvian}}]{Gour2018NC}%
  \BibitemOpen
  \bibfield  {author} {\bibinfo {author} {\bibfnamefont {G.}~\bibnamefont {Gour}}, \bibinfo {author} {\bibfnamefont {D.}~\bibnamefont {Jennings}}, \bibinfo {author} {\bibfnamefont {F.}~\bibnamefont {Buscemi}}, \bibinfo {author} {\bibfnamefont {R.}~\bibnamefont {Duan}},\ and\ \bibinfo {author} {\bibfnamefont {I.}~\bibnamefont {Marvian}},\ }\bibfield  {title} {\bibinfo {title} {Quantum majorization and a complete set of entropic conditions for quantum thermodynamics},\ }\href {https://doi.org/10.1038/s41467-018-06261-7} {\bibfield  {journal} {\bibinfo  {journal} {Nat. Commun.}\ }\textbf {\bibinfo {volume} {9}},\ \bibinfo {pages} {5352} (\bibinfo {year} {2018})}\BibitemShut {NoStop}%
\bibitem [{\citenamefont {Theurer}\ \emph {et~al.}(2023)\citenamefont {Theurer}, \citenamefont {Zanoni}, \citenamefont {Maria~Scandolo},\ and\ \citenamefont {Gour}}]{Theurer2023NJP}%
  \BibitemOpen
  \bibfield  {author} {\bibinfo {author} {\bibfnamefont {T.}~\bibnamefont {Theurer}}, \bibinfo {author} {\bibfnamefont {E.}~\bibnamefont {Zanoni}}, \bibinfo {author} {\bibfnamefont {C.}~\bibnamefont {Maria~Scandolo}},\ and\ \bibinfo {author} {\bibfnamefont {G.}~\bibnamefont {Gour}},\ }\bibfield  {title} {\bibinfo {title} {Thermodynamic state convertibility is determined by qubit cooling and heating},\ }\href {https://doi.org/10.1088/1367-2630/ad0d38} {\bibfield  {journal} {\bibinfo  {journal} {New J. Phys.}\ }\textbf {\bibinfo {volume} {25}},\ \bibinfo {pages} {123017} (\bibinfo {year} {2023})}\BibitemShut {NoStop}%
\bibitem [{\citenamefont {Gour}(2022)}]{Gour2022PRXQ}%
  \BibitemOpen
  \bibfield  {author} {\bibinfo {author} {\bibfnamefont {G.}~\bibnamefont {Gour}},\ }\bibfield  {title} {\bibinfo {title} {Role of quantum coherence in thermodynamics},\ }\href {https://doi.org/10.1103/PRXQuantum.3.040323} {\bibfield  {journal} {\bibinfo  {journal} {PRX Quantum}\ }\textbf {\bibinfo {volume} {3}},\ \bibinfo {pages} {040323} (\bibinfo {year} {2022})}\BibitemShut {NoStop}%
\bibitem [{\citenamefont {Brand\~ao}\ \emph {et~al.}(2013)\citenamefont {Brand\~ao}, \citenamefont {Horodecki}, \citenamefont {Oppenheim}, \citenamefont {Renes},\ and\ \citenamefont {Spekkens}}]{Brandao2013PRL}%
  \BibitemOpen
  \bibfield  {author} {\bibinfo {author} {\bibfnamefont {F.~G. S.~L.}\ \bibnamefont {Brand\~ao}}, \bibinfo {author} {\bibfnamefont {M.}~\bibnamefont {Horodecki}}, \bibinfo {author} {\bibfnamefont {J.}~\bibnamefont {Oppenheim}}, \bibinfo {author} {\bibfnamefont {J.~M.}\ \bibnamefont {Renes}},\ and\ \bibinfo {author} {\bibfnamefont {R.~W.}\ \bibnamefont {Spekkens}},\ }\bibfield  {title} {\bibinfo {title} {Resource theory of quantum states out of thermal equilibrium},\ }\href {https://doi.org/10.1103/PhysRevLett.111.250404} {\bibfield  {journal} {\bibinfo  {journal} {Phys. Rev. Lett.}\ }\textbf {\bibinfo {volume} {111}},\ \bibinfo {pages} {250404} (\bibinfo {year} {2013})}\BibitemShut {NoStop}%
\bibitem [{Foo()}]{Footnote1}%
  \BibitemOpen
  \href@noop {} {}\bibinfo {note} {Note that $W(\rho,H)$ is the average work that can be extracted in the limit of infinitely many copies of $\rho$~\cite{Brandao2013PRL}, or in the limit of infinitely many steps in the protocol~\cite{Skrzypczyk2014NC}.}\BibitemShut {Stop}%
\bibitem [{\citenamefont {Umegaki}(1962)}]{Umegaki1962}%
  \BibitemOpen
  \bibfield  {author} {\bibinfo {author} {\bibfnamefont {H.}~\bibnamefont {Umegaki}},\ }\bibfield  {title} {\bibinfo {title} {Conditional expectation in an operator algebra, {IV} ({E}ntropy {A}nd {I}nformation)},\ }\href {https://doi.org/10.2996/kmj/1138844604} {\bibfield  {journal} {\bibinfo  {journal} {Kodai Math. Sem. Rep.}\ }\textbf {\bibinfo {volume} {14}},\ \bibinfo {pages} {59} (\bibinfo {year} {1962})}\BibitemShut {NoStop}%
\bibitem [{\citenamefont {Oppenheim}\ \emph {et~al.}(2002)\citenamefont {Oppenheim}, \citenamefont {Horodecki}, \citenamefont {Horodecki},\ and\ \citenamefont {Horodecki}}]{Oppenheim2002PRL}%
  \BibitemOpen
  \bibfield  {author} {\bibinfo {author} {\bibfnamefont {J.}~\bibnamefont {Oppenheim}}, \bibinfo {author} {\bibfnamefont {M.}~\bibnamefont {Horodecki}}, \bibinfo {author} {\bibfnamefont {P.}~\bibnamefont {Horodecki}},\ and\ \bibinfo {author} {\bibfnamefont {R.}~\bibnamefont {Horodecki}},\ }\bibfield  {title} {\bibinfo {title} {Thermodynamical approach to quantifying quantum correlations},\ }\href {https://doi.org/10.1103/PhysRevLett.89.180402} {\bibfield  {journal} {\bibinfo  {journal} {Phys. Rev. Lett.}\ }\textbf {\bibinfo {volume} {89}},\ \bibinfo {pages} {180402} (\bibinfo {year} {2002})}\BibitemShut {NoStop}%
\bibitem [{\citenamefont {Alimuddin}\ \emph {et~al.}(2020)\citenamefont {Alimuddin}, \citenamefont {Guha},\ and\ \citenamefont {Parashar}}]{Alimuddin2020PRE}%
  \BibitemOpen
  \bibfield  {author} {\bibinfo {author} {\bibfnamefont {M.}~\bibnamefont {Alimuddin}}, \bibinfo {author} {\bibfnamefont {T.}~\bibnamefont {Guha}},\ and\ \bibinfo {author} {\bibfnamefont {P.}~\bibnamefont {Parashar}},\ }\bibfield  {title} {\bibinfo {title} {Independence of work and entropy for equal-energetic finite quantum systems: Passive-state energy as an entanglement quantifier},\ }\href {https://doi.org/10.1103/PhysRevE.102.012145} {\bibfield  {journal} {\bibinfo  {journal} {Phys. Rev. E}\ }\textbf {\bibinfo {volume} {102}},\ \bibinfo {pages} {012145} (\bibinfo {year} {2020})}\BibitemShut {NoStop}%
\bibitem [{\citenamefont {Alimuddin}\ \emph {et~al.}(2019)\citenamefont {Alimuddin}, \citenamefont {Guha},\ and\ \citenamefont {Parashar}}]{Alimuddin2019PRA}%
  \BibitemOpen
  \bibfield  {author} {\bibinfo {author} {\bibfnamefont {M.}~\bibnamefont {Alimuddin}}, \bibinfo {author} {\bibfnamefont {T.}~\bibnamefont {Guha}},\ and\ \bibinfo {author} {\bibfnamefont {P.}~\bibnamefont {Parashar}},\ }\bibfield  {title} {\bibinfo {title} {Bound on ergotropic gap for bipartite separable states},\ }\href {https://doi.org/10.1103/PhysRevA.99.052320} {\bibfield  {journal} {\bibinfo  {journal} {Phys. Rev. A}\ }\textbf {\bibinfo {volume} {99}},\ \bibinfo {pages} {052320} (\bibinfo {year} {2019})}\BibitemShut {NoStop}%
\bibitem [{\citenamefont {Vaccaro}\ \emph {et~al.}(2008)\citenamefont {Vaccaro}, \citenamefont {Anselmi}, \citenamefont {Wiseman},\ and\ \citenamefont {Jacobs}}]{Vaccaro2008PRA}%
  \BibitemOpen
  \bibfield  {author} {\bibinfo {author} {\bibfnamefont {J.~A.}\ \bibnamefont {Vaccaro}}, \bibinfo {author} {\bibfnamefont {F.}~\bibnamefont {Anselmi}}, \bibinfo {author} {\bibfnamefont {H.~M.}\ \bibnamefont {Wiseman}},\ and\ \bibinfo {author} {\bibfnamefont {K.}~\bibnamefont {Jacobs}},\ }\bibfield  {title} {\bibinfo {title} {Tradeoff between extractable mechanical work, accessible entanglement, and ability to act as a reference system, under arbitrary superselection rules},\ }\href {https://doi.org/10.1103/PhysRevA.77.032114} {\bibfield  {journal} {\bibinfo  {journal} {Phys. Rev. A}\ }\textbf {\bibinfo {volume} {77}},\ \bibinfo {pages} {032114} (\bibinfo {year} {2008})}\BibitemShut {NoStop}%
\bibitem [{\citenamefont {Kwon}\ \emph {et~al.}(2018)\citenamefont {Kwon}, \citenamefont {Jeong}, \citenamefont {Jennings}, \citenamefont {Yadin},\ and\ \citenamefont {Kim}}]{Kwon2018PRL}%
  \BibitemOpen
  \bibfield  {author} {\bibinfo {author} {\bibfnamefont {H.}~\bibnamefont {Kwon}}, \bibinfo {author} {\bibfnamefont {H.}~\bibnamefont {Jeong}}, \bibinfo {author} {\bibfnamefont {D.}~\bibnamefont {Jennings}}, \bibinfo {author} {\bibfnamefont {B.}~\bibnamefont {Yadin}},\ and\ \bibinfo {author} {\bibfnamefont {M.~S.}\ \bibnamefont {Kim}},\ }\bibfield  {title} {\bibinfo {title} {Clock--work trade-off relation for coherence in quantum thermodynamics},\ }\href {https://doi.org/10.1103/PhysRevLett.120.150602} {\bibfield  {journal} {\bibinfo  {journal} {Phys. Rev. Lett.}\ }\textbf {\bibinfo {volume} {120}},\ \bibinfo {pages} {150602} (\bibinfo {year} {2018})}\BibitemShut {NoStop}%
\bibitem [{\citenamefont {Vinjanampathy}\ and\ \citenamefont {Anders}(2016)}]{Vinjanampathy2016CP}%
  \BibitemOpen
  \bibfield  {author} {\bibinfo {author} {\bibfnamefont {S.}~\bibnamefont {Vinjanampathy}}\ and\ \bibinfo {author} {\bibfnamefont {J.}~\bibnamefont {Anders}},\ }\bibfield  {title} {\bibinfo {title} {Quantum thermodynamics},\ }\href {https://doi.org/10.1080/00107514.2016.1201896} {\bibfield  {journal} {\bibinfo  {journal} {Contemp. Phys.}\ }\textbf {\bibinfo {volume} {57}},\ \bibinfo {pages} {545} (\bibinfo {year} {2016})}\BibitemShut {NoStop}%
\bibitem [{\citenamefont {Ciampini}\ \emph {et~al.}(2017)\citenamefont {Ciampini}, \citenamefont {Mancino}, \citenamefont {Orieux}, \citenamefont {Vigliar}, \citenamefont {Mataloni}, \citenamefont {Paternostro},\ and\ \citenamefont {Barbieri}}]{Ciampini2017npjQI}%
  \BibitemOpen
  \bibfield  {author} {\bibinfo {author} {\bibfnamefont {M.~A.}\ \bibnamefont {Ciampini}}, \bibinfo {author} {\bibfnamefont {L.}~\bibnamefont {Mancino}}, \bibinfo {author} {\bibfnamefont {A.}~\bibnamefont {Orieux}}, \bibinfo {author} {\bibfnamefont {C.}~\bibnamefont {Vigliar}}, \bibinfo {author} {\bibfnamefont {P.}~\bibnamefont {Mataloni}}, \bibinfo {author} {\bibfnamefont {M.}~\bibnamefont {Paternostro}},\ and\ \bibinfo {author} {\bibfnamefont {M.}~\bibnamefont {Barbieri}},\ }\bibfield  {title} {\bibinfo {title} {Experimental extractable work-based multipartite separability criteria},\ }\href {https://doi.org/10.1038/s41534-017-0011-9} {\bibfield  {journal} {\bibinfo  {journal} {npj Quantum Inf.}\ }\textbf {\bibinfo {volume} {3}},\ \bibinfo {pages} {10} (\bibinfo {year} {2017})}\BibitemShut {NoStop}%
\bibitem [{\citenamefont {Francica}\ \emph {et~al.}(2017)\citenamefont {Francica}, \citenamefont {Goold}, \citenamefont {Plastina},\ and\ \citenamefont {Paternostro}}]{Francica2017npjQI}%
  \BibitemOpen
  \bibfield  {author} {\bibinfo {author} {\bibfnamefont {G.}~\bibnamefont {Francica}}, \bibinfo {author} {\bibfnamefont {J.}~\bibnamefont {Goold}}, \bibinfo {author} {\bibfnamefont {F.}~\bibnamefont {Plastina}},\ and\ \bibinfo {author} {\bibfnamefont {M.}~\bibnamefont {Paternostro}},\ }\bibfield  {title} {\bibinfo {title} {Daemonic ergotropy: enhanced work extraction from quantum correlations},\ }\href {https://doi.org/10.1038/s41534-017-0012-8} {\bibfield  {journal} {\bibinfo  {journal} {npj Quantum Inf.}\ }\textbf {\bibinfo {volume} {3}},\ \bibinfo {pages} {12} (\bibinfo {year} {2017})}\BibitemShut {NoStop}%
\bibitem [{\citenamefont {Andolina}\ \emph {et~al.}(2019)\citenamefont {Andolina}, \citenamefont {Keck}, \citenamefont {Mari}, \citenamefont {Campisi}, \citenamefont {Giovannetti},\ and\ \citenamefont {Polini}}]{Andolina2019PRL}%
  \BibitemOpen
  \bibfield  {author} {\bibinfo {author} {\bibfnamefont {G.~M.}\ \bibnamefont {Andolina}}, \bibinfo {author} {\bibfnamefont {M.}~\bibnamefont {Keck}}, \bibinfo {author} {\bibfnamefont {A.}~\bibnamefont {Mari}}, \bibinfo {author} {\bibfnamefont {M.}~\bibnamefont {Campisi}}, \bibinfo {author} {\bibfnamefont {V.}~\bibnamefont {Giovannetti}},\ and\ \bibinfo {author} {\bibfnamefont {M.}~\bibnamefont {Polini}},\ }\bibfield  {title} {\bibinfo {title} {Extractable work, the role of correlations, and asymptotic freedom in quantum batteries},\ }\href {https://doi.org/10.1103/PhysRevLett.122.047702} {\bibfield  {journal} {\bibinfo  {journal} {Phys. Rev. Lett.}\ }\textbf {\bibinfo {volume} {122}},\ \bibinfo {pages} {047702} (\bibinfo {year} {2019})}\BibitemShut {NoStop}%
\bibitem [{\citenamefont {Monsel}\ \emph {et~al.}(2020)\citenamefont {Monsel}, \citenamefont {Fellous-Asiani}, \citenamefont {Huard},\ and\ \citenamefont {Auff\`eves}}]{Monsel2020PRL}%
  \BibitemOpen
  \bibfield  {author} {\bibinfo {author} {\bibfnamefont {J.}~\bibnamefont {Monsel}}, \bibinfo {author} {\bibfnamefont {M.}~\bibnamefont {Fellous-Asiani}}, \bibinfo {author} {\bibfnamefont {B.}~\bibnamefont {Huard}},\ and\ \bibinfo {author} {\bibfnamefont {A.}~\bibnamefont {Auff\`eves}},\ }\bibfield  {title} {\bibinfo {title} {The energetic cost of work extraction},\ }\href {https://doi.org/10.1103/PhysRevLett.124.130601} {\bibfield  {journal} {\bibinfo  {journal} {Phys. Rev. Lett.}\ }\textbf {\bibinfo {volume} {124}},\ \bibinfo {pages} {130601} (\bibinfo {year} {2020})}\BibitemShut {NoStop}%
\bibitem [{\citenamefont {Opatrn\'y}\ \emph {et~al.}(2021)\citenamefont {Opatrn\'y}, \citenamefont {Misra},\ and\ \citenamefont {Kurizki}}]{Opatrny2021PRL}%
  \BibitemOpen
  \bibfield  {author} {\bibinfo {author} {\bibfnamefont {T.}~\bibnamefont {Opatrn\'y}}, \bibinfo {author} {\bibfnamefont {A.}~\bibnamefont {Misra}},\ and\ \bibinfo {author} {\bibfnamefont {G.}~\bibnamefont {Kurizki}},\ }\bibfield  {title} {\bibinfo {title} {Work generation from thermal noise by quantum phase-sensitive observation},\ }\href {https://doi.org/10.1103/PhysRevLett.127.040602} {\bibfield  {journal} {\bibinfo  {journal} {Phys. Rev. Lett.}\ }\textbf {\bibinfo {volume} {127}},\ \bibinfo {pages} {040602} (\bibinfo {year} {2021})}\BibitemShut {NoStop}%
\bibitem [{\citenamefont {Yang}\ \emph {et~al.}(2023)\citenamefont {Yang}, \citenamefont {Yang}, \citenamefont {Alimuddin}, \citenamefont {Salvia}, \citenamefont {Fei}, \citenamefont {Zhao}, \citenamefont {Nimmrichter},\ and\ \citenamefont {Luo}}]{Yang2023PRL}%
  \BibitemOpen
  \bibfield  {author} {\bibinfo {author} {\bibfnamefont {X.}~\bibnamefont {Yang}}, \bibinfo {author} {\bibfnamefont {Y.-H.}\ \bibnamefont {Yang}}, \bibinfo {author} {\bibfnamefont {M.}~\bibnamefont {Alimuddin}}, \bibinfo {author} {\bibfnamefont {R.}~\bibnamefont {Salvia}}, \bibinfo {author} {\bibfnamefont {S.-M.}\ \bibnamefont {Fei}}, \bibinfo {author} {\bibfnamefont {L.-M.}\ \bibnamefont {Zhao}}, \bibinfo {author} {\bibfnamefont {S.}~\bibnamefont {Nimmrichter}},\ and\ \bibinfo {author} {\bibfnamefont {M.-X.}\ \bibnamefont {Luo}},\ }\bibfield  {title} {\bibinfo {title} {Battery capacity of energy-storing quantum systems},\ }\href {https://doi.org/10.1103/PhysRevLett.131.030402} {\bibfield  {journal} {\bibinfo  {journal} {Phys. Rev. Lett.}\ }\textbf {\bibinfo {volume} {131}},\ \bibinfo {pages} {030402} (\bibinfo {year} {2023})}\BibitemShut {NoStop}%
\bibitem [{\citenamefont {Datta}\ \emph {et~al.}(2023)\citenamefont {Datta}, \citenamefont {Ganardi}, \citenamefont {Kondra},\ and\ \citenamefont {Streltsov}}]{Datta2023PRL}%
  \BibitemOpen
  \bibfield  {author} {\bibinfo {author} {\bibfnamefont {C.}~\bibnamefont {Datta}}, \bibinfo {author} {\bibfnamefont {R.}~\bibnamefont {Ganardi}}, \bibinfo {author} {\bibfnamefont {T.~V.}\ \bibnamefont {Kondra}},\ and\ \bibinfo {author} {\bibfnamefont {A.}~\bibnamefont {Streltsov}},\ }\bibfield  {title} {\bibinfo {title} {Is there a finite complete set of monotones in any quantum resource theory?},\ }\href {https://doi.org/10.1103/PhysRevLett.130.240204} {\bibfield  {journal} {\bibinfo  {journal} {Phys. Rev. Lett.}\ }\textbf {\bibinfo {volume} {130}},\ \bibinfo {pages} {240204} (\bibinfo {year} {2023})}\BibitemShut {NoStop}%
\bibitem [{\citenamefont {Hardy}\ \emph {et~al.}(1988)\citenamefont {Hardy}, \citenamefont {Littlewood},\ and\ \citenamefont {Pólya}}]{Hardy-Littlewood-Pólya}%
  \BibitemOpen
  \bibfield  {author} {\bibinfo {author} {\bibfnamefont {G.~H.}\ \bibnamefont {Hardy}}, \bibinfo {author} {\bibfnamefont {J.~E.}\ \bibnamefont {Littlewood}},\ and\ \bibinfo {author} {\bibfnamefont {G.}~\bibnamefont {Pólya}},\ }\href@noop {} {\emph {\bibinfo {title} {Inequalities}}},\ \bibinfo {edition} {2nd}\ ed.\ (\bibinfo  {publisher} {Cambridge University Press},\ \bibinfo {year} {1988})\BibitemShut {NoStop}%
\bibitem [{Note3()}]{Note3}%
  \BibitemOpen
  \bibinfo {note} {{To see this, it suffices to combine Lemmas 6, 10, 16 and Definition 15 in Ref.~\cite {Purity-review}.}}\BibitemShut {Stop}%
\bibitem [{\citenamefont {Sion}(1958)}]{Sion1958}%
  \BibitemOpen
  \bibfield  {author} {\bibinfo {author} {\bibfnamefont {M.}~\bibnamefont {Sion}},\ }\bibfield  {title} {\bibinfo {title} {On general minimax theorems},\ }\href {https://doi.org/10.2140/pjm.1958.8.171} {\bibfield  {journal} {\bibinfo  {journal} {Pac. J. Math.}\ }\textbf {\bibinfo {volume} {8}},\ \bibinfo {pages} {171} (\bibinfo {year} {1958})}\BibitemShut {NoStop}%
\bibitem [{\citenamefont {Navascu\'es}\ and\ \citenamefont {Popescu}(2014)}]{Navascues2014PRL}%
  \BibitemOpen
  \bibfield  {author} {\bibinfo {author} {\bibfnamefont {M.}~\bibnamefont {Navascu\'es}}\ and\ \bibinfo {author} {\bibfnamefont {S.}~\bibnamefont {Popescu}},\ }\bibfield  {title} {\bibinfo {title} {How energy conservation limits our measurements},\ }\href {https://doi.org/10.1103/PhysRevLett.112.140502} {\bibfield  {journal} {\bibinfo  {journal} {Phys. Rev. Lett.}\ }\textbf {\bibinfo {volume} {112}},\ \bibinfo {pages} {140502} (\bibinfo {year} {2014})}\BibitemShut {NoStop}%
\bibitem [{\citenamefont {Nielsen}(1999)}]{Nielsen1999PRL}%
  \BibitemOpen
  \bibfield  {author} {\bibinfo {author} {\bibfnamefont {M.~A.}\ \bibnamefont {Nielsen}},\ }\bibfield  {title} {\bibinfo {title} {Conditions for a class of entanglement transformations},\ }\href {https://doi.org/10.1103/PhysRevLett.83.436} {\bibfield  {journal} {\bibinfo  {journal} {Phys. Rev. Lett.}\ }\textbf {\bibinfo {volume} {83}},\ \bibinfo {pages} {436} (\bibinfo {year} {1999})}\BibitemShut {NoStop}%
\bibitem [{\citenamefont {Allahverdyan}\ \emph {et~al.}(2004)\citenamefont {Allahverdyan}, \citenamefont {Balian},\ and\ \citenamefont {Nieuwenhuizen}}]{Allahverdyan2004EPL}%
  \BibitemOpen
  \bibfield  {author} {\bibinfo {author} {\bibfnamefont {A.~E.}\ \bibnamefont {Allahverdyan}}, \bibinfo {author} {\bibfnamefont {R.}~\bibnamefont {Balian}},\ and\ \bibinfo {author} {\bibfnamefont {T.~M.}\ \bibnamefont {Nieuwenhuizen}},\ }\bibfield  {title} {\bibinfo {title} {Maximal work extraction from finite quantum systems},\ }\href {https://doi.org/10.1209/epl/i2004-10101-2} {\bibfield  {journal} {\bibinfo  {journal} {Europhys. Lett.}\ }\textbf {\bibinfo {volume} {67}},\ \bibinfo {pages} {565} (\bibinfo {year} {2004})}\BibitemShut {NoStop}%
\bibitem [{\citenamefont {Tirone}\ \emph {et~al.}(2024)\citenamefont {Tirone}, \citenamefont {Salvia}, \citenamefont {Chessa},\ and\ \citenamefont {Giovannetti}}]{Tirone2024SciPostPhys}%
  \BibitemOpen
  \bibfield  {author} {\bibinfo {author} {\bibfnamefont {S.}~\bibnamefont {Tirone}}, \bibinfo {author} {\bibfnamefont {R.}~\bibnamefont {Salvia}}, \bibinfo {author} {\bibfnamefont {S.}~\bibnamefont {Chessa}},\ and\ \bibinfo {author} {\bibfnamefont {V.}~\bibnamefont {Giovannetti}},\ }\bibfield  {title} {\bibinfo {title} {{Quantum work capacitances: Ultimate limits for energy extraction on noisy quantum batteries}},\ }\href {https://doi.org/10.21468/SciPostPhys.17.2.041} {\bibfield  {journal} {\bibinfo  {journal} {SciPost Phys.}\ }\textbf {\bibinfo {volume} {17}},\ \bibinfo {pages} {041} (\bibinfo {year} {2024})}\BibitemShut {NoStop}%
\bibitem [{\citenamefont {Rosset}\ \emph {et~al.}(2020)\citenamefont {Rosset}, \citenamefont {Schmid},\ and\ \citenamefont {Buscemi}}]{Rosset2020PRL}%
  \BibitemOpen
  \bibfield  {author} {\bibinfo {author} {\bibfnamefont {D.}~\bibnamefont {Rosset}}, \bibinfo {author} {\bibfnamefont {D.}~\bibnamefont {Schmid}},\ and\ \bibinfo {author} {\bibfnamefont {F.}~\bibnamefont {Buscemi}},\ }\bibfield  {title} {\bibinfo {title} {Type-independent characterization of spacelike separated resources},\ }\href {https://doi.org/10.1103/PhysRevLett.125.210402} {\bibfield  {journal} {\bibinfo  {journal} {Phys. Rev. Lett.}\ }\textbf {\bibinfo {volume} {125}},\ \bibinfo {pages} {210402} (\bibinfo {year} {2020})}\BibitemShut {NoStop}%
\bibitem [{\citenamefont {Watrous}(2018)}]{Watrous-book}%
  \BibitemOpen
  \bibfield  {author} {\bibinfo {author} {\bibfnamefont {J.}~\bibnamefont {Watrous}},\ }\href {https://doi.org/10.1017/9781316848142} {\emph {\bibinfo {title} {The Theory of Quantum Information}}}\ (\bibinfo  {publisher} {Cambridge University Press},\ \bibinfo {year} {2018})\BibitemShut {NoStop}%
\bibitem [{\citenamefont {Regula}\ \emph {et~al.}(2021)\citenamefont {Regula}, \citenamefont {Takagi},\ and\ \citenamefont {Gu}}]{Regula2021Quantum}%
  \BibitemOpen
  \bibfield  {author} {\bibinfo {author} {\bibfnamefont {B.}~\bibnamefont {Regula}}, \bibinfo {author} {\bibfnamefont {R.}~\bibnamefont {Takagi}},\ and\ \bibinfo {author} {\bibfnamefont {M.}~\bibnamefont {Gu}},\ }\bibfield  {title} {\bibinfo {title} {Operational applications of the diamond norm and related measures in quantifying the non-physicality of quantum maps},\ }\href {https://doi.org/10.22331/q-2021-08-09-522} {\bibfield  {journal} {\bibinfo  {journal} {{Quantum}}\ }\textbf {\bibinfo {volume} {5}},\ \bibinfo {pages} {522} (\bibinfo {year} {2021})}\BibitemShut {NoStop}%
\end{thebibliography}%

\clearpage
\CY{\section{End Matter}}

\CY{{\bf\em Appendix A: On the convexity and compactness of $\freeset$---}When the working hypotheses hold for $\freeop$, a convex mixture of two state-preparation channels of any two free states must also be in $\freeop$.
Since allowed operations cannot output resourceful states from free inputs [i.e., Eq.~\eqref{Eq: golden rule}], this convex mixture must also be a state-preparation channel of some free state.
This implies that $\freeset$ is convex.

To see $\freeset$'s compactness, suppose we have a limit point $\eta$ of $\freeset$ in the topology induced by the one norm $\norm{\cdot}_1$, where, for an operator $P$, its {\em one norm} reads  $\norm{P}_1\coloneqq{\rm tr}(|P|)$~\cite{QIC-book}.
This means there is a countable sequence $\{\eta_n\}_{n=1}^\infty\subseteq\freeset$ achieving $\lim_{n\to\infty}\norm{\eta-\eta_n}_1=0$.
Let $\Lambda_n(\cdot)\coloneqq\eta_n{\rm tr}(\cdot)$ and $\Lambda_\infty(\cdot)\coloneqq\eta{\rm tr}(\cdot)$ be the corresponding state-preparation channels.
Now, recall that, for any Hermitian-preserving map $\Lambda$ with input system $S$, its {\em diamond norm}~\cite{Watrous-book} is defined by 
\begin{align}\label{Eq: diamon norm}
\norm{\Lambda}_\diamond\coloneqq\max\left\{\norm{(\Lambda\otimes\mathcal{I})(\rho)}_1\,|\,\text{$\rho$: state in $SA$}\right\}, 
\end{align}
where $\mathcal{I}$ is the identity map acting on an auxiliary system $A$ with the same dimension as $S$ (see, e.g., Ref.~\cite{Regula2021Quantum}).
Then,
\begin{align}
\lim_{n\to\infty}\norm{\Lambda_\infty-\Lambda_n}_\diamond&=\lim_{n\to\infty}\max_{\rho_A}\norm{(\eta-\eta_n)\otimes\rho_A}_1\nonumber\\
&=\lim_{n\to\infty}\norm{\eta-\eta_n}_1=0,
\end{align}
where $\rho_A$ is a state living in the auxiliary system $A$.
Consequently, $\Lambda_\infty$ is a limit point of $\freeop$ in the topology induced by $\norm{\cdot}_\diamond$.
By the working hypotheses, $\freeop$ is closed in this topology, meaning that $\Lambda_\infty\in\freeop$.
Since $\Lambda_\infty$ is the state-preparation channel of $\eta$, and allowed operations cannot output resourceful states from free inputs [i.e., Eq.~\eqref{Eq: golden rule}], we conclude that $\eta\in\freeset$.
Namely, $\freeset$ contains all its limit points in the topology induced by $\norm{\cdot}_1$, meaning that it is closed in this topology.
Because $\freeset$ is a subset of quantum states, it is bounded. 
Hence, it is bounded and closed in a finite-dimensional space, implying that it is indeed compact in the desired topology.
}

{\bf\em Appendix B: Proof of Theorem~\ref{Result:conversion}---}If $\mathcal{N}(\rho)=\sigma$ for some $\mathcal{N}\in\freeop$, then, by definition, $\Delta_{\freeop}(\rho,H)\ge\Delta_{\freeop}(\sigma,H)$ $\forall\,H$.
Hence, it suffices to prove the opposite direction. 
For $\rho$ and $\sigma$, suppose we have Eq.~\eqref{Eq:Result:conversion}.
This means, for every $\epsilon\id\le H\le\delta\id$,
\begin{align}\label{Eq:computation001}
0\le\max_{\mathcal{N}\in\freeop}\Big(\Delta\left[\mathcal{N}(\rho),H\right]-\Delta(\sigma,H)\Big).
\end{align}
Now, note that, for any Hamiltonian $H$ and state $\rho$ in a $d$-dimensional system, we have the following formula:
\begin{align}\label{Eq:WgapRelation}
&\frac{\Wdiff\left(\rho,H\right)}{k_BT\ln2} = \frac{{\rm tr}\left(H\rho\right)}{k_BT\ln2} + \log_2{\rm tr}\left(e^{-\frac{H}{k_BT}}\right) - \log_2d.
\end{align}
Using the above formula in Eq.~\eqref{Eq:computation001}, we obtain
\mbox{$
0\le\max_{\mathcal{N}\in\freeop}{\rm tr}\Big(\left[\mathcal{N}(\rho)-\sigma\right]H\Big).
$}
Minimising over all Hamiltonians $H$ with $\epsilon\id\le H\le\delta\id$ gives
\begin{align}\label{Eq: for mean energy proof}
0\le\min_{\epsilon\id\le H\le\delta\id}\max_{\mathcal{N}\in\freeop}{\rm tr}\Big(\left[\mathcal{N}(\rho)-\sigma\right]H\Big).
\end{align}
Now, note that both $\freeop$ and $\{H\,|\,\epsilon\id\le H\le\delta\id\}$ are compact and convex sets.
Moreover, for fixed $\rho$ and $\sigma$, define
$
f(\mathcal{N},H)\coloneqq{\rm tr}\Big(\left[\mathcal{N}(\rho)-\sigma\right]H\Big).
$
Then, for any fixed $\mathcal{N}$, the function $f$ is linear in $H$, meaning that it is continuous as well as quasi-convex in $H$.
Similarly, for any fixed $H$, the function $f$ is also linear in $\mathcal{N}$, meaning that it is continuous as well as quasi-concave in $\mathcal{N}$.
The above implies that Sion's minimax theorem~\cite{Sion1958} can be applied to switch the order of the maximisation and minimisation, giving
\begin{align}\label{Eq:Thm1-computation001-3}
0\le\max_{\mathcal{N}\in\freeop}\min_{\epsilon\id\le H\le\delta\id}{\rm tr}\Big(\left[\mathcal{N}(\rho)-\sigma\right]H\Big).
\end{align}
Define the function
\begin{align}
f_{\rm min}(\mathcal{N})\coloneqq\min_{\epsilon\id\le H\le\delta\id}{\rm tr}\Big(\left[\mathcal{N}(\rho)-\sigma\right]H\Big).
\end{align}
For two channels $\mathcal{E},\mathcal{L}$, without loss of generality, suppose $f_{\rm min}(\mathcal{E})\ge f_{\rm min}(\mathcal{L})$.
Then, we have (let $\widetilde{H}=H/\delta$)
\begin{align}
&\left|f_{\rm min}(\mathcal{E})-f_{\rm min}(\mathcal{L})\right|\le\delta\max_{(\epsilon/\delta)\id\le \widetilde{H}\le\id}{\rm tr}\Big(\left[\mathcal{E}(\rho)-\mathcal{L}(\rho)\right]\widetilde{H}\Big)\nonumber\\
&\le\delta\max_{0\le P\le\id}{\rm tr}\Big(\left[\mathcal{E}(\rho)-\mathcal{L}(\rho)\right]P\Big)\le\frac{\delta}{2}\norm{\mathcal{E}-\mathcal{L}}_\diamond,
\end{align}
meaning that $f_{\rm min}$ is Lipschitz continuous in the topology induced by the diamond norm $\norm{\cdot}_\diamond$ [Eq.~\eqref{Eq: diamon norm}].
Here, we have used the trace distance formula $\norm{\rho-\sigma}_1/2 = \max_{0\le P\le\id}{\rm tr}[(\rho-\sigma)P]$~\cite{QIC-book}.
Since $\freeop$ is compact in the same topology 
(by working hypothesis), there is $\mathcal{N}_*\in\freeop$ achieving $f_{\rm min}(\mathcal{N}_*)=\max_{\mathcal{N}\in\freeop}f_{\rm min}(\mathcal{N})$.
Thus, Eq.~\eqref{Eq:Thm1-computation001-3} leads to
\begin{align}
0\le\min_{\epsilon\id\le H\le\delta\id}{\rm tr}\Big(\left[\mathcal{N}_*(\rho)-\sigma\right]H\Big),
\end{align}
which is equivalent to $0\le\mathcal{N}_*(\rho)-\sigma.$
This means $\mathcal{N}_*(\rho)-\sigma$ is a trace-less positive semi-definite operator, i.e., it is zero.
Consequently, $\mathcal{N}_*(\rho)=\sigma$, which concludes the proof.
\hfill$\square$

\CY{
{\bf\em Appendix C: Proof of Theorem~\ref{coro}---}Since there is always a unitary converting $\rho$ to $\rho^\downarrow$,  $\rho\stackrel{\rm MU}{\longrightarrow}\sigma$ if and only if $\rho^\downarrow\stackrel{\rm MU}{\longrightarrow}\sigma^\downarrow$, which holds if and only if (by setting $\epsilon=0$ in Theorem~\ref{Result:conversion})
\begin{align}\label{Eq: info computation001}
\Delta_{\rm MU}(\rho^\downarrow,H)\ge\Delta_{\rm MU}(\sigma^\downarrow,H)\quad\forall\,0\le H\le\delta\id,
\end{align}
where we set ``$\freeop={\rm MU}$'', the set of all mixed unitary channels.
Using Eq.~\eqref{Eq:WgapRelation}, we learn that Eq.~\eqref{Eq: info computation001} holds
if and only if
$
\max_{\mathcal{N}\in{\rm MU}}{\rm tr}\left[\mathcal{N}(\rho^\downarrow)H\right]\ge\max_{\mathcal{N}\in{\rm MU}}{\rm tr}\left[\mathcal{N}(\sigma^\downarrow)H\right].
$
One can check that $\max_{\mathcal{N}\in{\rm MU}}{\rm tr}\left[\mathcal{N}(\rho^\downarrow)H\right]=\max_U{\rm tr}\left(U\rho^\downarrow U^\dagger H\right)$ $\forall\,H\ge0$ (``$\max_U$'' denotes maximisation over all unitaries $U$).
Hence, Eq.~\eqref{Eq: info computation001} holds if and only if
\begin{align}\label{Eq: info computation002}
\max_{U}{\rm tr}\left(U\rho^\downarrow U^\dagger H\right)\ge\max_{U}{\rm tr}\left(U\sigma^\downarrow U^\dagger H\right)\quad\forall\,0\le H\le\delta\id.
\end{align}
Let $\{E_i^\downarrow(H)\}_i$ be $H$'s eigen-energies  in the {\em non-increasing} order; i.e., \mbox{$ E_i^\downarrow(H)\ge E_{i+1}^\downarrow(H)$} \mbox{$\forall\,i=0,...,d-2$}.
In Lemma~\ref{lemma}, we show that Eq.~\eqref{Eq: info computation002} holds [and so does Eq.~\eqref{Eq: info computation001}] if and only if
\begin{align}\label{Eq: passive state induced conditions}
\sum_{i=0}^{d-1}\left(\rho_i^\downarrow-\sigma_i^\downarrow\right)E_i^\downarrow(H)\ge0\quad\forall\,0\le H\le\delta\id.
\end{align}

Now, recall that $H_k\coloneqq\omega\sum_{i=0}^k\proj{i}$ and $\{\ket{i}\}_i$ is the pre-fixed basis and $\omega>0$.
Using Eq.~\eqref{Eq:WgapRelation} again, we conclude that Eq.~\eqref{Eq:Info non-equilibrium result} (i.e., Theorem~\ref{coro}'s statement) holds if and only if
\begin{align}\label{Eq: k conditions}
\omega\sum_{i=0}^{k}\left(\rho_i^\downarrow-\sigma_i^\downarrow\right)\ge0\quad\forall\,k=0,...,d-2.
\end{align}
It remains to show that Eqs.~\eqref{Eq: passive state induced conditions} and~\eqref{Eq: k conditions} are equivalent.
As Eq.~\eqref{Eq: passive state induced conditions} implies Eq.~\eqref{Eq: k conditions}, it remains to show the converse.
For a Hamiltonian $0\le H\le\delta\id$, suppose $E_N^\downarrow(H)>0$ is $H$'s smallest positive eigen-energy.
Then, we have
\begin{align}
&\sum_{i=0}^{d-1}\left(\rho_i^\downarrow-\sigma_i^\downarrow\right)E_i^\downarrow(H)\nonumber\\
&\quad\quad\ge E_1^\downarrow(H)\sum_{i=0}^1\left(\rho_i^\downarrow-\sigma_i^\downarrow\right)+\sum_{i=2}^{d-1}\left(\rho_i^\downarrow-\sigma_i^\downarrow\right)E_i^\downarrow(H)\nonumber\\
&\quad\quad\ge E_2^\downarrow(H)\sum_{i=0}^2\left(\rho_i^\downarrow-\sigma_i^\downarrow\right)+\sum_{i=3}^{d-1}\left(\rho_i^\downarrow-\sigma_i^\downarrow\right)E_i^\downarrow(H)\nonumber\\
&\quad\quad...\ge E_N^\downarrow(H)\sum_{i=0}^N\left(\rho_i^\downarrow-\sigma_i^\downarrow\right)\ge0.
\end{align}
The second line is due to $\rho_0^\downarrow-\sigma_0^\downarrow\ge0$ [setting $k=0$ in Eq.~\eqref{Eq: k conditions}] and $E_0^\downarrow\ge E_1^\downarrow$.
The third line is due to $\sum_{i=0}^1\left(\rho_i^\downarrow-\sigma_i^\downarrow\right)\ge0$ [setting $k=1$ in Eq.~\eqref{Eq: k conditions}] and $E_1^\downarrow\ge E_2^\downarrow$.
Applying this argument $N$ times gives the last line, which is non-negative by setting $k=N$ in Eq.~\eqref{Eq: k conditions}.
This shows that Eq.~\eqref{Eq: k conditions} implies Eq.~\eqref{Eq: passive state induced conditions}.
Finally, since $\omega>0$, Eq.~\eqref{Eq: k conditions} 
is equivalent to \mbox{$\rho\succ\sigma$} [defined in Eq.~\eqref{Eq:majorisation}], leading to Eq.~\eqref{Eq:Q2ndLaws info}.
\hfill$\square$
}

\CY{
Finally, it remains to prove the following lemma:
\begin{lemma}\label{lemma} For $\rho$ and $0\le H\le\delta\id$, we have that
\begin{align}
\max_{U}{\rm tr}\left(U\rho^\downarrow U^\dagger H\right)=\sum_{i=0}^{d-1}\rho_i^\downarrow E_i^\downarrow(H).
\end{align}
\end{lemma}
\begin{proof}
For $0\le H\le\delta\id$, let $\widetilde{H}\coloneqq\delta\id-H$.
Then we have 
\begin{align}\label{Eq: change of Hamiltonian}
\max_{U}{\rm tr}\left(U\rho^\downarrow U^\dagger H\right)
=\delta-\min_U{\rm tr}\left(U\rho^\downarrow U^\dagger \widetilde{H}\right).
\end{align}
Let us write
$
H = \sum_{i=0}^{d-1} E_i^\downarrow(H)\proj{E_i}
$,
where $\ket{E_i}$ is $H$'s energy eigenstate with the eigen-energy $ E_i^\downarrow(H)$.
Hence, \mbox{$\widetilde{H}=\sum_{i=0}^{d-1} \left(\delta-E_i^\downarrow(H)\right)\proj{E_i}$} with its eigen-energies in the {\em non-decreasing} order.
Using the property of passive states~\cite{Allahverdyan2004EPL} (see also Eqs.~(1) to (3) in Ref.~\cite{Perarnau-LlobetPRX2015}), we obtain
\begin{align}
\min_U{\rm tr}\left(U\rho^\downarrow U^\dagger \widetilde{H}\right)={\rm tr}\left(\widetilde{\rho}_{\rm passive}\widetilde{H}\right),
\end{align}
where $\widetilde{\rho}_{\rm passive}\coloneqq\sum_{i=0}^{d-1}\rho_i^\downarrow\proj{E_i}$ is a so-called passive state with the Hamiltonian $\widetilde{H}$ (i.e., diagonal in $\widetilde{H}$'s eigenbasis with non-increasing eigenvalues with respect to $\widetilde{H}$'s eigen-energies).
Together with Eq.~\eqref{Eq: change of Hamiltonian}, we have
$
\max_{U}{\rm tr}\left(U\rho^\downarrow U^\dagger H\right)
={\rm tr}\left(\widetilde{\rho}_{\rm passive}H\right)=\sum_{i=0}^{d-1}\rho_i^\downarrow E_i^\downarrow(H).
$
\end{proof}
}

\CY{
{\bf\em Appendix D: Proof of Theorem~\ref{coro:state}---}Define the set
\begin{align}
\mathcal{O}_{R}^{\rm min}
\coloneqq{\rm Cov}\{\text{state-preparation of $\freeset$ \& identity}\},
\end{align}
which is the set containing the identity channel, all state-preparation channels that prepare some free state, and their convex mixture.
Since, in our approach, $\freeset$ is always convex and compact (see Appendix A), $\mathcal{O}_{R}^{\rm min}$ is a set satisfying the working hypothesis [to see its compactness, it suffices to consider the continuous function $f(p,\eta)\coloneqq p\eta{\rm tr}(\cdot)+(1-p)\mathcal{I}(\cdot)$ and note that it maps the compact set $[0,1]\times\freeset$ onto $\mathcal{O}_{R}^{\rm min}$].
Consider energy scales \mbox{$0\le\epsilon<\delta$} and a given $\rho\notin\freeset$.
For a given free state $\eta_*\in\freeset$, since no channel in $\mathcal{O}_R^{\rm min}$ can convert it to $\rho$, Theorem~\ref{Result:conversion} implies that there is a Hamiltonian $\epsilon\id\le H\le\delta\id$ achieving
$\Delta_{\mathcal{O}_R^{\rm min}}(\rho,H)>\Delta_{\mathcal{O}_R^{\rm min}}(\eta_*,H)$.
Now, the structure of $\mathcal{O}_R^{\rm min}$ ensures
$
\Delta_{\mathcal{O}_R^{\rm min}}(\eta_*,H)=\max_{\eta\in\freeset}\Delta(\eta,H).
$
Hence,
\begin{align}\label{Eq: O min strict ineq}
\Delta_{\mathcal{O}_R^{\rm min}}(\rho,H)>\max_{\eta\in\freeset}\Delta(\eta,H).
\end{align}
Finally, using Eq.~\eqref{Eq:WgapRelation}, we have
\begin{align}
&\Delta_{\mathcal{O}_R^{\rm min}}(\rho,H)=\max_{\substack{0\le p\le1;\\\eta\in\freeset}}\Big(p\Delta\left(\rho,H\right)+(1-p)\Delta\left(\eta,H\right)\Big)\nonumber\\
&\quad\quad=\max_{0\le p\le1}\Big(p\Delta\left(\rho,H\right)+(1-p)\max_{\eta\in\freeset}\Delta(\eta,H)\Big).
\end{align}
Suppose $p_*$ achieves the above maximum.
Combining with Eq.~\eqref{Eq: O min strict ineq}, we obtain
$
p_*\max_{\eta\in\freeset}\Delta(\eta,H)<p_*\Delta\left(\rho,H\right).
$
This thus means we must have $p_*>0$.
Dividing $p_*$ on both sides concludes the proof.
\hfill$\square$
}

\CY{As a remark, such results are usually proved by the hyperplane separation theorem (see, e.g., Ref.~\cite {Takagi2019}). Here, we prove it by Sion's minimax theorem (via Theorem~\ref{Result:conversion}).}

\CY{
{\bf\em Appendix E: Proof of Theorem~\ref{result:quantification R}---}Since $\freeop$ contains all state-preparation channels of free states (from the working hypothesis), Eq.~\eqref{Eq:Def Delta_OR} implies that $\Delta_{\freeop}(\kappa,H)=\max_{\eta\in\freeset}\Delta(\eta,H)$
$\forall\,\epsilon\id\le H\le\delta\id$ and $\forall\,\kappa\in\freeset$, meaning that $\mathcal{M}_{\freeop}(\eta)=0$ $\forall\,\eta\in\freeset$.
Also, for $\rho\notin\freeset$, Theorem~\ref{coro:state} implies that there is $\epsilon\id\le H\le\delta\id$ achieving
$
\max_{\eta\in\freeset}\Delta_{\freeop}(\eta,H)<\Delta_{\freeop}(\rho,H)
$.
Hence, $\mathcal{M}_{\freeop}(\rho)>0$ $\forall\,\rho\notin\freeset$.
Finally, to show that $\mathcal{M}_{\freeop}$ is non-increasing under allowed operations, note that
$
\Delta_{\freeop}[\mathcal{N}(\rho),H]\le\Delta_{\freeop}(\rho,H)
$
$\forall\,\rho$ and $\forall\,\mathcal{N}\in\freeop$.
Hence,
\begin{align}
&\mathcal{M}_{\freeop}[\mathcal{N}(\rho)]\coloneqq
\max_{\epsilon\id\le H\le\delta\id}\left(\Delta_{\freeop}[\mathcal{N}(\rho),H] - \max_{\eta\in\freeset}\Delta_{\freeop}(\eta,H)\right)\nonumber\\
&\le
\max_{\epsilon\id\le H\le\delta\id}\left[\Delta_{\freeop}(\rho,H) - \max_{\eta\in\freeset}\Delta_{\freeop}(\eta,H)\right]=\mathcal{M}_{\freeop}(\rho),
\end{align}
which concludes the proof.
\hfill$\square$
}

\end{document}